\documentclass[sigconf, nonacm]{acmart}

\usepackage{amsmath,amsfonts}
\usepackage{graphicx}
\usepackage{textcomp}
\usepackage[usenames, dvipsnames]{xcolor}
\usepackage{booktabs}
\usepackage{pgfplots}
\usepackage{subcaption}
\usepackage{colortbl}
\usepackage{bbding}
\usepackage{braket}
\usepackage{epsfig,latexsym,graphics}
\usepackage[font={small}]{caption}
\usepackage{setspace}
\usepackage{enumitem}
\pgfplotsset{compat=1.16}
\usepackage{cmd}
\usepackage{svg}
\usepackage[a-2b]{pdfx}

\newcommand\vldbdoi{XX.XX/XXX.XX}
\newcommand\vldbpages{XXX-XXX}
\newcommand\vldbvolume{14}
\newcommand\vldbissue{1}
\newcommand\vldbyear{2020}
\newcommand\vldbauthors{\authors}
\newcommand\vldbtitle{\shorttitle} 
\newcommand\vldbavailabilityurl{https://github.com/mingyu-hkustgz/CubeGraph}
\newcommand\vldbpagestyle{plain}

\begin{document}

\title{CubeGraph: Efficient Retrieval-Augmented Generation for Spatial and Temporal Data}

\author{Mingyu Yang}
\affiliation{%
  \institution{HKUST (GZ) \& HKUST}
  \country{China}
}
\email{myang250@connect.hkust-gz.edu.cn}

\author{Wentao Li}
\affiliation{%
\institution{University of Leicester}
  \country{United Kingdom}
}
\email{wl226@leicester.ac.uk}

\author{Wei Wang}
\affiliation{%
  \institution{HKUST (GZ) \& HKUST}
  \country{China}
}
\email{weiwcs@ust.hk}

\renewcommand{\shortauthors}{Trovato et al.}

\begin{abstract}
Hybrid queries combining high-dimensional vector similarity search with spatio-temporal filters are increasingly critical for modern retrieval-augmented generation (RAG) systems. Existing systems typically handle these workloads by nesting vector indices within low-dimensional spatial structures, such as R-trees. However, this decoupled architecture fragments the vector space, forcing the query engine to invoke multiple disjoint sub-indices per query. This fragmentation destroys graph routing connectivity, incurs severe traversal overhead, and struggles to optimize for complex spatial boundaries. In this paper, we propose CubeGraph, a novel indexing framework designed to natively integrate vector search with arbitrary spatial constraints. CubeGraph partitions the spatial domain using a hierarchical grid, maintaining modular vector graphs within each cell. During query execution, CubeGraph dynamically stitches together adjacent cube-level indices on the fly whenever their spatial cells intersect with the query filter. This dynamic graph integration restores global connectivity, enabling a unified, single-pass nearest-neighbor traversal that eliminates the overhead of fragmented sub-index invocations. Extensive evaluations on real-world datasets demonstrate that CubeGraph significantly outperforms state-of-the-art baselines, offering superior query execution performance, scalability, and flexibility for complex hybrid workloads.
\end{abstract}

\maketitle

\pagestyle{\vldbpagestyle}
\begingroup\small\noindent\raggedright\textbf{PVLDB Reference Format:}\\
\vldbauthors. \vldbtitle. PVLDB, \vldbvolume(\vldbissue): \vldbpages, \vldbyear.\\
\href{https://doi.org/\vldbdoi}{doi:\vldbdoi}
\endgroup
\begingroup
\renewcommand\thefootnote{}\footnote{\noindent
This work is licensed under the Creative Commons BY-NC-ND 4.0 International License. Visit \url{https://creativecommons.org/licenses/by-nc-nd/4.0/} to view a copy of this license. For any use beyond those covered by this license, obtain permission by emailing \href{mailto:info@vldb.org}{info@vldb.org}. Copyright is held by the owner/author(s). Publication rights licensed to the VLDB Endowment. \\
\raggedright Proceedings of the VLDB Endowment, Vol. \vldbvolume, No. \vldbissue\ %
ISSN 2150-8097. \\
\href{https://doi.org/\vldbdoi}{doi:\vldbdoi} \\
}\addtocounter{footnote}{-1}\endgroup

\ifdefempty{\vldbavailabilityurl}{}{
\vspace{.3cm}
\begingroup\small\noindent\raggedright\textbf{PVLDB Artifact Availability:}\\
The source code, data, and/or other artifacts have been made available at \url{\vldbavailabilityurl}.
\endgroup
}

\section{Introduction}
Modern unstructured data, such as text, code, images, user profiles, and videos, is intrinsically linked with spatial and temporal metadata, such as geolocations, timestamps, and trajectories~\cite{spatio-temporal-survey-SIGMOD-2022}. High-dimensional vector embeddings have become the standard foundation for retrieval augmented generation over such data, mapping each object into a latent space where semantic or visual relevance is measured via inner product or Euclidean distance. Beyond pure similarity search, emerging applications increasingly demand hybrid queries that jointly evaluate vector similarity alongside complex spatial and temporal predicates. In practice, these predicates extend far beyond simple bounding boxes; they require retrieval engines to satisfy constraints over irregular geographic polygons, spatial intersections, temporal windows, and dynamically changing spatio-temporal conditions. While supporting such hybrid queries enables richer information retrieval, it introduces severe challenges for efficient query processing. We illustrate this need through the following motivating examples.

\stitle{Motivating Example 1: Urban event retrieval.}
A city management platform stores geo-tagged reports, images, and videos, each represented by a vector embedding and associated with a timestamp and location (Fig.~\ref{fig:Motivation}). Given a query such as \textit{flooded streets}, an analyst may search for semantically similar records within a specified time window, but strictly limited to objects located inside an irregular flood-impact region. This query necessitates the seamless integration of vector similarity search with complex spatial and temporal filters.

\stitle{Motivating Example 2: Regional multimedia retrieval.}
A multimedia database manages geo-tagged images and videos alongside their vector embeddings and timestamps (Fig.~\ref{fig:Motivation}). Given a query image, a user may search for the most similar objects captured during a specific time interval, constrained within a complex polygonal region and excluding restricted subregions. This requires jointly processing vector similarity with non-rectangular spatial filtering and temporal constraints.

\begin{figure}[!t]
    \centering
    \includegraphics[width=\linewidth]{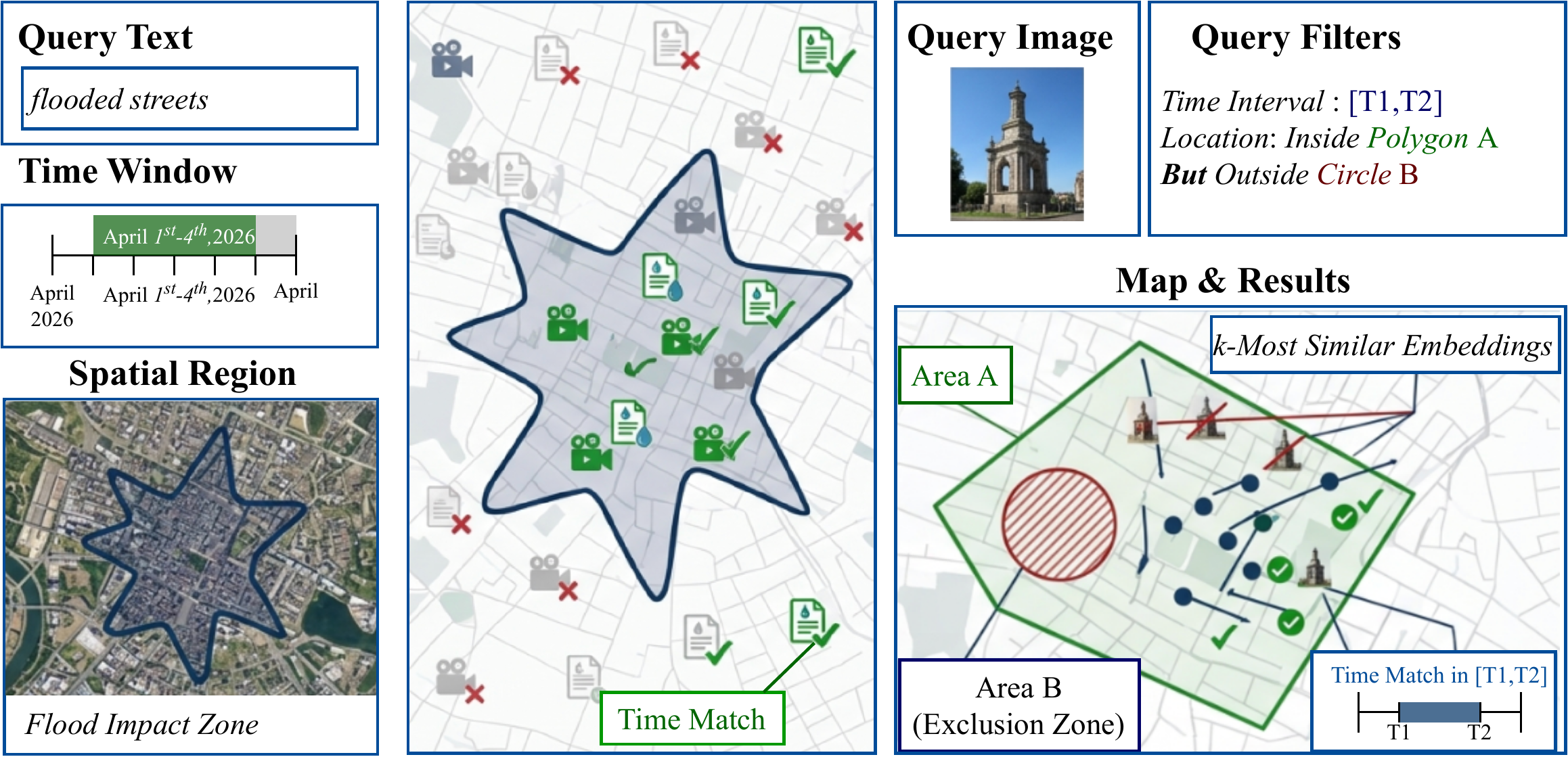}
    \caption{Motivating examples for hybrid vector similarity search with spatio-temporal filters.
    Each data object is represented by a high-dimensional vector embedding \(\mathbf{x}_i\) and associated with spatio-temporal metadata \(\mathbf{s}_i\) (e.g., geolocation and timestamp).
    (1) Urban event retrieval: a city management platform stores geo-tagged reports, images, and videos. Given a query vector \(\mathbf{q}\) (e.g., \textit{flooded streets}) and a spatio-temporal filter \(\phi\), the task is to retrieve the top-\(k\) most similar objects satisfying \(\phi(\mathbf{s}_i) = 1\).
    (2) Regional multimedia retrieval: a multimedia database stores geo-tagged images and videos. Given a query image \(\mathbf{q}\) and a complex polygonal filter \(\phi\), the system must return the top-\(k\) nearest neighbors within the specified spatio-temporal constraints.}\label{fig:Motivation}
\end{figure}

Existing database solutions typically model this as a filtered approximate \(k\)-nearest neighbor (AKNN) search problem. Since the exact result is hard to obtain in the high-dimensional space due to the curse of dimensionality~\cite {Curse-of-dim-1998}.  Among various vector index~\cite{OPQ-PAMI-2014,PQ-fast-scan-VLDB-2015,LSH-1999,AQ-2014-cvpr,Rabitq-SIGMOD-2024,MRQ-RESQ-arxiv-2024-Mingyu,PDX-SIGMOD-2025}, graph-based vector indices (e.g., HNSW) offer superior naive AKNN query efficiency~\cite{ann-benchmakrs,NSG-VLDB-2019-deng-cai,HNSW-PAMI-2020,Acorn-SIGMOD-2024,ANNSurvey-TKDE-2020-Wei-Wang,HVS-VLDB-2021-kejing-lu,tMRNG:journals/pacmmod/PengCCYX23,CPG-SIGMOD-2026-Shangqi-Lu,Diskann-NIPS-2019,DEG-SIGMOD-2025,PDX-SIGMOD-2025,HQI-SIGMOD-2023,Graph-ANNS-Survey-VLDB-2021-mengzhao,Starling-SIGMOD-2024-Mengzhao,Graph-Index-Accelearte-SIGMOD-2025-Mengzhao,ADSampling:journals/sigmod/GaoL23,BSA-DDC-ICDE-2024-Mingyu}, most existing methods attempt to adapt the graph structure or traversal process to accommodate metadata filters~\cite{Acorn-SIGMOD-2024,AnalyticDB-VLDB-2020,iRangeGraph-SIGMOD-2025,ESG-arxiv-2025-Mingyu,Filtered-diskann-WWW-2023,Window-Filter-ICML-2024,SeRF-SIGMOD-2024,ELI-VLDB-2026-Mingyu,RangePQ-SIGMOD-2025-fangyuan-sibo,Wow-Range-SIGMOD-2025,DIGRA-SIGMOD-2025-sibo-menxu-cuhk,Beyond-Vector-Search-Jiadong-Xie-SIGMOD-2025,VSAG-arxiv-2025,Milvus-SIGMOD-2021,E2E-Arxiv-2026-Mingyu,UNG-SIGMOD-2025,Vector-database-Liguoliang-2024-SIGMOD}. A standard graph index navigates the high-dimensional space by routing through proximity-based neighbor connections until converging on the $k$-nearest neighbors of query. Early execution paradigms, namely $\PRE$ and $\POST$, leave the underlying graph topology unchanged but modify the traversal logic. $\PRE$ prunes nodes failing the predicate during traversal, which severely degrades graph connectivity and search accuracy under low filter selectivity. Conversely, $\POST$ traverses the graph ignoring the filter and verifies predicates post-hoc, leading to massive redundant distance computations and degraded efficiency.

\stitle{Challenge.}
To overcome the limitations of single-graph paradigms, state-of-the-art methods adopt a decoupled architecture, nesting multiple vector graphs within a low-dimensional spatial index. For instance, to handle 1D metadata (e.g., a time window), methods like WindowFilter~\cite{Window-Filter-ICML-2024}, iRange~\cite{iRangeGraph-SIGMOD-2025}, ESG~\cite{ESG-arxiv-2025-Mingyu}, and WoW~\cite{Wow-Range-SIGMOD-2025} build segment-tree-like structures where each tree node maintains a separate graph index. Consequently, answering a range filter query introduces an \(O(\log N)\) multiplicative overhead to the query execution time, alongside an \(O(N \log N)\) space complexity since the graph index has linear complexity.
In spatio-temporal scenarios, metadata dimensionality inherently exceeds 1D (e.g., 2D geolocation plus a 1D timestamp). Furthermore, spatial filters are rarely regular rectangles; they are often complex polygons or circles. To handle multi-dimensional metadata, a straightforward extension is to organize graph indices using multi-dimensional trees, such as R-trees or KD-trees~\cite{Mesh-VLDBJ-2025,KHI-arxiv-2026}. During query execution, the engine traverses the tree and invokes the graph indices attached to the overlapping tree nodes, subsequently merging the results~\cite{Mesh-VLDBJ-2025}. However, even with only 2D metadata, a KD-tree requires \(O(\sqrt{N})\) tree nodes to reconstruct a given 2D rectangular filter. This architectural decoupling forces the query engine to invoke a massive number of disjoint graph subqueries. This subquery explosion destroys the global routing connectivity of the vector space, severely limiting the search performance, flexibility, and scalability of existing tree-graph methods.

\stitle{Our Idea.}
In this paper, we propose \(\CG\), a novel index structure and search paradigm designed to efficiently process hybrid queries with arbitrary spatio-temporal filters. Our fundamental insight is that a graph remains navigable across multi-dimensional boundaries as long as the underlying index bounds its adjacent neighbors and supports dynamic connectivity among them. To this end, \(\CG\) constructs a hierarchical grid over the spatio-temporal metadata space, mapping data points into localized, modular cube indices. By leveraging multiple levels of spatial granularity, the hierarchical grid guarantees that the number of involved graph indices is tightly controlled, effectively eliminating the exponential subquery explosion inherent to KD-tree or R-tree-based methods.

Crucially, rather than treating these cubes as isolated search domains, \(\CG\) introduces a lightweight, on-the-fly graph stitching mechanism. During query execution, as the search algorithm identifies the bounded set of cubes intersecting the query filter, it dynamically links the nodes to adjacent cubes, achieving a merged graph index. This creates a unified, query-specific routing graph in real-time. This dynamic integration avoids the massive redundant distance computations of \(\POST\), the graph disconnection issues of \(\PRE\), and the traversal overhead of decoupled tree-graph architectures. In addition, users can adjust the number of layers and granularity of $\CG$ according to the query workload to further improve efficiency.

\begin{table}[!t]
\begin{small}
  \centering
\caption{Comparison of \(\CG\) with existing filtered vector search paradigms.
\underline{Performance} measures overall accuracy and query efficiency under complex spatial constraints.
\underline{Compact} indicates whether the method avoids the index size explosion when handling multi-dimensional spatial filters.
\underline{Connectivity} captures the ability to maintain a unified, well-connected graph during search, avoiding fragmentation.
\underline{Flexibility} assesses support for arbitrary spatio-temporal filter shapes (e.g., polygons, irregular regions).
}\vspace{-2ex}
\resizebox{8.5cm}{!}{
\begin{tabular}{l|c|cccc}
    \toprule
    Feature       & \(\CG\) (Our)         & \(\POST\)          & \(\ACORN\)       & Tree-Graph    \\
    \midrule
    Performance   & \(\star\star\star\) & \(\star\)        & \(\star\)      & \(\star\star\)  \\
    Compact        & \checkmark          & \checkmark       & \checkmark     & \(\times\)    \\
    Connectivity   & \checkmark          & \checkmark       & \(\times\)     & \(\times\)    \\
    Flexibility    & \checkmark          & \checkmark       & \checkmark     & \(\times\)    \\
    \bottomrule
  \end{tabular}}
  \label{tab:features}
\end{small}
\end{table}

\stitle{Contributions.}
We summarize our main contributions as follows:

\sstitle{Problem Analysis.}
We conduct an in-depth analysis of vector similarity search under complex spatio-temporal constraints. We identify that the multi-dimensionality of metadata and the geometric complexity of query filters lead to a subquery explosion in existing decoupled architectures, acting as the primary bottleneck for scalability and flexibility.

\sstitle{The \(\CG\) Framework.}
Based on the insight that dynamically stitched sub-graphs can achieve routing performance comparable to a monolithic index, we propose \(\CG\). This hierarchical grid index framework natively supports complex spatial query filters, offering high query execution efficiency while maintaining a strictly bounded index space overhead.

\sstitle{Efficient Query Processing Strategies.} We design two optimized query execution strategies: a predetermined cube search for simple boundaries, and an on-the-fly merged search for complex, irregular geometries. This adaptive approach guarantees a bounded search space and ensures high throughput across diverse spatial filters.

\sstitle{Extensive Evaluation.}
We design a comprehensive suite of experiments utilizing real-world and synthetic datasets under diverse query workloads. Our evaluations demonstrate that \(\CG\) achieves highly stable performance across varying spatial constraints, delivering up to a 5\(\times\) speedup over state-of-the-art baselines.

Due to space constraints, some proofs and experiments are omitted, which can be found in our technical report~\cite{technicalreport}.  

\section{Preliminary}\label{sec:preliminary}

In this section, we formally define the hybrid search problem over vector similarity and spatio-temporal metadata in \S~\ref{sec:prob-defn}. Following the definitions, we review related work on existing solutions in \S~\ref{sec:prior-works}. Finally, we introduce the index merging techniques utilized within our framework in \S~\ref{sec:merge}.

\subsection{Problem Definition}\label{sec:prob-defn}

We first introduce the basic notations and then formally define the core problem studied in this paper.

\stitle{Definition of a Data Point}.
A \emph{data point} is defined as a tuple \(o=(\mathbf{x}, \mathbf{s})\), where \(o\) is a unique object identifier, \(\mathbf{x} \in \mathbb{R}^d\) is its \(d\)-dimensional vector embedding, and \(\mathbf{s} \in \mathbb{R}^m\) represents its \(m\)-dimensional spatio-temporal metadata (e.g., longitude, latitude, timestamp).

\stitle{Definition of a Dataset}.
A \emph{dataset} is a collection \(\mathcal{D} = \{( \mathbf{x}_i, \mathbf{s}_i)\}\) comprising \(i\in [1,N]\) data points.

\stitle{Definition of a Spatio-Temporal Filter}.
A \emph{spatio-temporal filter} is a predicate \(\phi: \mathbb{R}^m \to \{0,1\}\) defined over the metadata space. A data point \(o_i\) \emph{satisfies} \(\phi\) if and only if \(\phi(\mathbf{s}_i) = 1\). This filter can represent axis-aligned rectangles, complex polygons, temporal windows, or any arbitrary intersection and union of these constraints.

\stitle{Definition of Filtered Candidate Set}
Given a dataset \(\mathcal{D}\) and a filter \(\phi\), the \emph{filtered candidate set} is defined as \(\mathcal{D}_\phi = \{(\mathbf{x}_i, \mathbf{s}_i) \in \mathcal{D} \mid \phi(\mathbf{s}_i) = 1\}\).

\stitle{Hybrid AKNN Query}
Given a query vector \(\mathbf{q} \in \mathbb{R}^d\), a spatio-temporal filter \(\phi\), and an integer \(k\), a \emph{hybrid approximate \(k\)-nearest neighbor (AKNN) query} returns a result set \(\mathcal{R} \subseteq \mathcal{D}_\phi\) with \(|\mathcal{R}| = k\), such that the vectors in \(\mathcal{R}\) are approximately the \(k\) closest to \(\mathbf{q}\) under Euclidean distance (or inner product) among all valid points in \(\mathcal{D}_\phi\).

\noindent Table~\ref{tab:notation} summarizes the key notations used throughout this paper.

\subsection{Existing Solutions}\label{sec:prior-works}

We briefly review graph-based approximate nearest neighbor search (ANNS) methods and existing strategies for handling filters.

\stitle{Graph-based ANNS.}
Proximity graph methods represent the dataset as a graph \(G = (V, E)\), where each node \(v \in V\) corresponds to a data point, and edges connect approximate nearest neighbors in the high-dimensional vector space. Query processing typically follows a greedy beam search: starting from a fixed entry node, the algorithm iteratively expands to the neighbors closest to \(\mathbf{q}\), terminating when a local minimum is reached. For example, $\HNSW$~\cite{HNSW-PAMI-2020} organizes nodes into a hierarchy of layers, achieving \(O(d \log N)\) query time with \(O(N)\) space complexity by perform edge occlusion. Subsequent work, such as NSG and tMNG, improved the occlusion strategy to further enhance efficiency~\cite{NSG-VLDB-2019-deng-cai,tMRNG:journals/pacmmod/PengCCYX23,CPG-SIGMOD-2026-Shangqi-Lu,Diskann-NIPS-2019,HVS-VLDB-2021-kejing-lu}.

\stitle{Pre-filtering and Post-filtering.}
The straightforward method to incorporate a filter \(\phi\) into graph search is Pre-filtering and Post-filtering~\cite{Filtered-diskann-WWW-2023}. \(\PRE\) actively skips nodes where \(\phi(\mathbf{s}_i) = 0\) during graph traversal. However, when filter selectivity is low (i.e., \(|\mathcal{D}_\phi| < |\mathcal{D}|\)), the effective routing subgraph becomes severely sparse and disconnected, leading to catastrophic recall degradation. Conversely, \(\POST\) traverses the full graph, ignoring the filter, applying \(\phi\) only to the final retrieved candidate set. While this preserves recall, it wastes massive computational resources calculating distances for unqualified nodes, becoming prohibitively slow when filter selectivity is high.

\stitle{Filtered Graph Index Methods.}
To overcome the limitations of naive filtering, several methods modify the graph structure to natively support predicates. Filtered-DiskANN~\cite{Filtered-diskann-WWW-2023} builds per-label subgraphs and stitches them together, though it is primarily designed for categorical label filters. $\NHQ$~\cite{NHQ-NIPS-2022-mengzhao-wang} constructs a hybrid graph encoding both vector proximity and attribute proximity edges, supporting structured and unstructured constraints. $\ACORN$~\cite{Acorn-SIGMOD-2024} augments $\HNSW$~with a dynamic neighbor expansion strategy during search to maintain connectivity under arbitrary predicates. However, these methods are primarily optimized for low-dimensional or categorical metadata and struggle to scale efficiently when confronted with multi-dimensional spatio-temporal filters, the critical gap we address in \S~\ref{sec:problem}.

\begin{table}[!t]
  \caption{Summary of Notations}\label{tab:notation}\vspace{-2ex}
  \small
  \begin{tabular*}{\linewidth}{@{\extracolsep{\fill}} p{18mm} | p{68mm}}
    \toprule
    Notation & Description \\
    \midrule
    \(\mathcal{D}\) & Dataset of \(N\) data points \\
    \(N\) & Number of data points in the dataset \\
    \(d\) & Dimensionality of vector embeddings \\
    \(m\) & Dimensionality of spatio-temporal metadata \\
    \(\mathbf{x}_i\) & Vector embedding of data point \(o_i\) \\
    \(\mathbf{s}_i\) & Spatio-temporal metadata of data point \(o_i\) \\
    \(\mathbf{q}\) & Query vector \\
    \(\phi\) & Spatio-temporal filter predicate \\
    \(k\) & Number of nearest neighbors to retrieve \\
    \(\mathcal{D}_\phi\) & Filtered candidate set satisfying \(\phi\) \\
    \bottomrule
  \end{tabular*}
\end{table}

\subsection{Index Merging}\label{sec:merge}
Vector index merging is a critical technique with wide applications in distributed systems, disk-based solutions, and heterogeneous computing architectures~\cite{Diskann-NIPS-2019,VSAG-arxiv-2025}. In scenarios where memory constraints prevent building a monolithic index from scratch across the entire dataset, systems must rely on constructing smaller sub-indices. Index merging techniques enable the consolidation of these pre-built sub-indices to achieve near-optimal search performance. However, naively querying a large number of disjoint sub-indices significantly degrades search efficiency.

To address this, overlapping-based methods, such as those used in DiskANN, force data points to be assigned to multiple partitions to maintain connectivity. More recent methods like FGIM~\cite{FGIM-Arxiv-2026-Antgroup} utilize an enhanced NN-Descent algorithm for fast index merging. Similarly, RNSM~\cite{RNSM-Arxiv-2026-Mingyu} identifies the nearest neighbors of partitioned data in other partitions—a crucial factor in determining the quality of the merged index. It also greedily selects pivots and reuses their search results to accelerate the merging process. Our proposed framework leverages a nearest-neighbor-based index merging approach~\cite{RNSM-Arxiv-2026-Mingyu}, which supports highly flexible merge operations while maintaining a low computational merge cost.

\section{Problem Analysis}\label{sec:problem}

In this section, we analyze the inherent limitations of existing tree-graph hybrid methods for spatio-temporal filtered approximate nearest neighbor search (ANNS) and motivate the design of our proposed framework, \CG.

\stitle{1D Range Filter Methods and Their Limitations.}
When the metadata is a scalar (e.g., a timestamp) and the filter is an interval \([l, r]\), state-of-the-art methods such as SeRF~\cite{SeRF-SIGMOD-2024}, WindowFilter~\cite{Window-Filter-ICML-2024}, iRange~\cite{iRangeGraph-SIGMOD-2025}, ESG~\cite{ESG-arxiv-2025-Mingyu}, and WoW~\cite{Wow-Range-SIGMOD-2025} organize graph indices using a compressed index or segment-tree-like structure. Each tree node covers a contiguous interval of the sorted metadata axis, and a query \([l, r]\) is decomposed into \(O(\log N)\) canonical nodes. The sub-graph of each canonical node is searched independently, and the results are subsequently merged. However, this architecture incurs an \(O(N \log N)\) space complexity (as each point is replicated across \(O(\log N)\) nodes) and requires \(O(k_{\text{sub}} \cdot \log N)\) graph searches per query, where \(k_{\text{sub}}\) is the search budget allocated per sub-graph.

\begin{obs}
For 1D range filters, tree-graph methods incur an \(O(\log N)\) multiplicative overhead in both storage space and query execution cost compared to a single monolithic graph index.
\end{obs}

\stitle{Multi-Dimensional Spatio-Temporal Filters.}
When the metadata is \(d\)-dimensional (e.g., a 2D geolocation combined with a timestamp yields \(d=3\)), a natural extension is to organize the sub-graph indices using a multi-dimensional spatial tree, such as an R-tree~\cite{R*-tree-SIGMOD-1990} or KD-tree~\cite{Revisit-KD-tree-2019-KDD}. A spatial query filter \(\phi\) (e.g., a rectangle or polygon) is processed by: (1) traversing the tree to identify all leaf nodes overlapping with \(\phi\); (2) executing a graph search on the sub-index of each overlapping leaf; and (3) merging the retrieved results. A fundamental result from computational geometry establishes that the number of KD-tree nodes overlapping a \(d\)-dimensional orthogonal range query is \(\Theta(N^{1-1/d})\) in the worst case.

\begin{obs}
For a 2D spatial filter, a KD-tree-based approach requires \(\Theta(\sqrt{N})\) sub-index invocations per query; for 3D spatio-temporal filters, this complexity grows to \(\Theta(N^{2/3})\). Given a dataset of \(N = 10^6\) points with 2D metadata, such a method invokes \(\sim 10^3\) independent sub-graph searches per query---each carrying its own beam-search initialization overhead---compared to a single graph search in pure ANNS. This \emph{subquery explosion} renders multi-dimensional tree-graph methods fundamentally impractical at scale.
\end{obs}

\stitle{Connectivity Degradation Under Filtering.}
Beyond the sheer volume of subqueries, these methods suffer from a secondary structural flaw: severe degradation of graph connectivity. When a filter \(\phi\) exhibits high selectivity, retaining only a small fraction of the total nodes (\(|\mathcal{D}_\phi| \ll N\)), the sub-graphs attached to individual tree nodes may contain very few qualifying points. Consequently, the greedy beam search within each isolated sub-graph becomes highly susceptible to getting trapped in local minima, leading to a drastic drop in recall. This issue is orthogonal to the subquery explosion: even if the number of invoked sub-indices is manageable, the individual sub-graphs often become too sparse to navigate reliably.

\begin{obs}
When filter selectivity is low (i.e., \(|\mathcal{D}_\phi|/N \ll 1\)), tree-graph methods suffer from a compounding effect of (a) an excessively high subquery count and (b) poor intra-subgraph connectivity, both of which severely degrade recall and search efficiency.
\end{obs}

\stitle{Motivation for $\CG$.}
The aforementioned observations establish a clear design imperative: an ideal system must bound the number of sub-indices invoked per query to \(O(1)\) regardless of the filter shape or dimensionality, while simultaneously preserving global graph connectivity across the filtered subset. $\CG$ addresses both challenges through two core innovations. First, a \emph{hierarchical grid index} partitions the metadata space into cubes at multiple granularities, strictly bounding the number of cubes involved in any query to a small constant per level---independent of \(N\) and the filter dimensionality. Second, a \emph{dynamic graph merging} mechanism fuses the graphs of adjacent cubes on the fly during query execution. This creates a unified routing graph that preserves the navigability and connectivity of the original proximity graph, even under highly restrictive filters. These two components are detailed in \S~\ref{sec:method}.

\section{Methodology}\label{sec:method}

In this section, we present the \CG{} framework in detail. We first introduce the hierarchical grid structure in~\S~\ref{sec:framework}. We then describe index construction in~\S~\ref{sec:index-construct} and query processing in~\S~\ref{sec:query}.

\subsection{CubeGraph Framework}\label{sec:framework}

\begin{figure*}[!t]
    \centering
    \includegraphics[width=\linewidth]{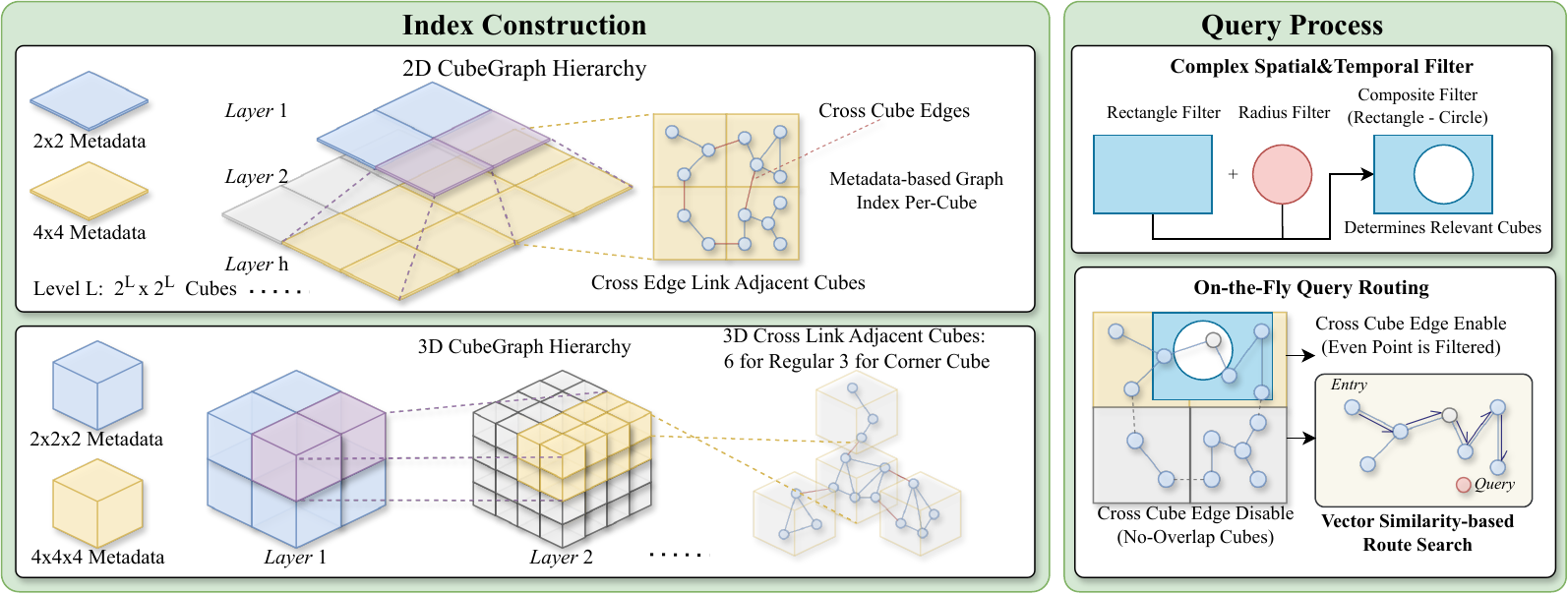}
    \caption{Framework of $\CG$. The hierarchical grid partitions the metadata space into multiple levels of cubes. At each level, data points are assigned to their containing cubes and local graph indexes are built. During query processing, cubes intersecting the filter are identified, and their local graphs are merged on-the-fly via cross-cube connections (dashed lines), forming a unified search graph. Note, we only plot the cross-cube edges of the boundary nodes for simplicity; the $\CG$ index requires that each node link cross-cube edges to the adjacent cube nearest neighbors, whether on the boundary of meta space or not.}
    \label{fig:framework}
\end{figure*}

\stitle{Framework Overview.}
Fig.~\ref{fig:framework} illustrates the $\CG$ framework. The hierarchical grid structure partitions the $m$-dimensional metadata space into multiple levels of granularity. At level $\ell$, the space is divided into $g^\ell$ uniform cubes. Each data point is assigned to exactly one cube per level based on its metadata $\mathbf{s}$. For each cube, we build a local graph index on the vectors of points in that cube. During query processing, given a filter $\phi$, we identify all cubes at each level that intersect $\phi$. For adjacent intersecting cubes, we dynamically add cross-links to adjacent cubes, effectively merging the local graphs into a unified search graph $\mathcal{G}^*$. This on-the-fly merging preserves the navigability of the graph across cube boundaries while keeping the search space bounded by the filter spatial extent.

\subsection{Index Construction}\label{sec:index-construct}
To construct the index $\kw{CI}$ of $\CG$, we proceed in two phases.
In Phase 1, we organize the dataset into a hierarchical grid across multiple layers. 
Each layer partitions the attribute space into disjoint cubes, and a local index (e.g., HNSW) is built for the points within each cube.
In Phase 2, we add cross-cube edges among points within the same layer to ensure global connectivity.
Since cubes are disjoint, both phases can be parallelized to improve efficiency.

\begin{algorithm}[t]
\small
\caption{Hierarchy Construction with Local Indexing}\label{alg:grid-construction}
\KwIn{Dataset $\mathcal{D}$, number of layers $L$, parameters $M$, $ef_{\text{c}}$}
\KwOut{Local index $\kw{CI}^{\text{Local}}$}
Load metadata and compute the global bounding box $\mathcal{B}$\;
\For{$\ell = 0$ \KwTo $L-1$}{
    Set granularity $g_\ell = 2^{\ell+1}$\;
    Partition $\mathcal{B}$ into $(g_\ell)^m$ cubes with side length $w_\ell = |\mathcal{B}| / g_\ell$\;
    \ForEach{$(\mathbf{x}_i, \mathbf{s}_i) \in \mathcal{D}$ \textbf{in parallel}}{
        Compute cube ID $c_i$ and assign point $i$ to cube $c_i$\;
    }
    \ForEach{non-empty cube $c$ \textbf{in parallel}}{
        Build an index over points in $c$ with parameters $M$ and $ef_{\text{c}}$\;
    }
}
\end{algorithm}

\stitle{Phase 1: Hierarchy Construction with Local Indexing.}
Algorithm~\ref{alg:grid-construction} describes Phase 1, which constructs the local index $\kw{CI}^{\text{Local}}$. The output $\kw{CI}^{\text{Local}}$ is a hierarchical structure of cubes across $L$ layers, where each cube maintains its own graph index.
Given a dataset $\mathcal{D}$ with $N$ points, we first load the attribute metadata and compute the global bounding box $\mathcal{B}$ over the attribute space (Line~1). 
We then build $L$ layers. For each layer $\ell \in [0, L-1]$, the attribute space is uniformly partitioned into $(2^{\ell+1})^m$ cubes, where $m$ is the attribute dimensionality (Line~3). 
Each cube has side length $w_\ell = |\mathcal{B}| / 2^{\ell+1}$ per dimension.
Each data point is assigned to its corresponding cube based on its attribute coordinates (Lines~5-6). 
For every non-empty cube $c$ at layer $\ell$, we construct a graph-based index (e.g., HNSW) using standard methods, with parameters $M$ (maximum degree) and $ef_{\text{c}}$ (construction beam width) (Lines~7-8).
Both point assignment and per-cube index construction are independent across cubes and can be parallelized (e.g., via OpenMP) for efficiency. During this process, we also build the cube adjacency structure: each cube has up to $2m$ face-adjacent neighbors, defined as cubes that differ by one unit in exactly one dimension and share an $(m-1)$-dimensional boundary. This adjacency information is used in the subsequent cross-cube edge construction phase.

\stitle{Phase 2: Cross-Cube Edge Addition.}
Note that the local index $\kw{CI}^{\text{Local}}$ consists of isolated indexes associated with cubes across the $L$ layers.
To establish global connectivity within each layer, we proceed to Phase 2, as described in Algorithm~\ref{alg:cross-edges}.
For each cube $c$ and its adjacent cubes at layer $\ell$, we add cross-cube edges from nodes in $c$ to nodes in each adjacent cube $c_{\text{adj}}$ (Lines~1--3).
Specifically, for each point $p \in c$, we initiate a search from the entry point of $c_{\text{adj}}$ using a beam search with width $ef_{\text{cross}}$ (Lines~4--5).
We then select the top $M_{\text{cross}}$ nearest neighbors from $c_{\text{adj}}$ and connect them to $p$ as cross-cube edges (Line~6).
These edges are stored separately from intra-cube edges in the graph.

\begin{algorithm}[t]
\small
\caption{Cross-Cube Edge Addition}\label{alg:cross-edges}
\KwIn{Local index $\kw{CI}^{\text{Local}}$, $M_{\text{cross}}$, $ef_{\text{cross}}$}
\KwOut{$\CG$ index $\kw{CI}$}
\ForEach{cube $c$ at each layer $\ell$}{
    \ForEach{point $p$ in cube $c$ \textbf{in parallel}}{
        \ForEach{adjacent cube $c_{\text{adj}}$ of $c$}{
            Use $p$ as the query\;
            $\mathcal{N} \gets$ perform AKNN search in $c_{\text{adj}}$ with $ef_{\text{cross}}$\;
            Select the top $M_{\text{cross}}$ nearest neighbors from $\mathcal{N}$\;
            Add these neighbors as cross-cube edges of $p$\;
        }
    }
}
\end{algorithm}

\begin{figure}
    \centering
    \includegraphics[width=0.85\linewidth]{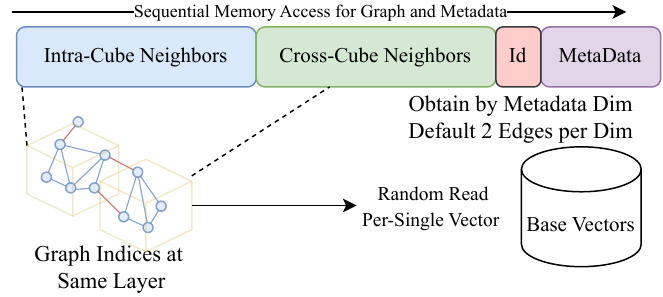}
    \vgap\caption{Memory layout of \(\CG\). Each node in \(\CG\) utilizes an identical memory layout. Unlike standard implementations such as hnswlib, our vector data is managed separately because the total number of nodes scales by a factor of \(L\) across \(L\) layers. Within any specific layer, a node maintains both intra-cube neighbors and cross-cube neighbors, the latter being determined by the metadata dimensionality. Additionally, aligning the graph structure with the metadata enables highly efficient filter evaluation during node traversal.}
    \label{fig:layout}
\end{figure}

\stitle{Design Rationale.}
The hierarchical structure with $L$ layers enables $\CG$ to adapt to filters of varying spatial extents. Coarse layers (small $\ell$) handle large filters efficiently with fewer, larger cubes, while fine layers (large $\ell$) provide precise filtering for small spatial regions. Cross-cube edges are essential for maintaining graph connectivity: without them, the search would be confined to a single cube, severely limiting recall. By connecting nodes of adjacent cubes, we create seamless routing paths that span multiple cubes while preserving the navigability properties of the underlying graph index. Fig.~\ref{fig:framework} illustrates how cross-cube edges (dashed lines) bridge local graphs during query processing. We also improve the memory layout for better search efficiency as detail in Fig.~\ref{fig:layout}.


\subsection{Query Processing}\label{sec:query}

Query processing in $\CG$ consists of two stages: (1) layer selection and cube identification, and (2) graph search with filtering. We present two query processing strategies tailored to different filter characteristics: a predetermined cube search for simple filters and an on-the-fly merged search for complex filters.

\stitle{Layer Selection Strategy.}
Given a filter $\phi$, we first select the appropriate layer $\ell^*$ for query execution. The key insight is to match the filter's characteristic length with the cube width at each layer. For an axis-aligned bounding box filter or convex hull, the characteristic length is the maximum side length; for a circular filter, it is the diameter. In practice, we select the layer $\ell^*$ by comparing the filter bounding box dimensions with the cube widths at each layer, choosing the layer where $w_\ell$ is closest to the filter characteristic length. More specifically, we find the layer with the largest cube width less than the characteristic length $r$ by binary search as the query layer where $r/2 < w_{l}<r$.

\stitle{Cube Identification.}
After selecting layer $\ell^*$, we identify all cubes intersecting the filter $\phi$. For simple filters (axis-aligned bounding boxes), we compute the cube IDs directly by discretizing the min/max bounds of filter. For complex filters (circles, polygons), we use a conservative approach: compute the filter's bounding box, identify all cubes intersecting this bounding box during graph search, then apply the filter predicate during search. Fig.~\ref{fig:framework} illustrates cube identification for different filter shapes.

\begin{algorithm}[t]
\caption{Predetermined Cube Search}\label{alg:predetermined-search}
\KwIn{Query $\mathbf{q}$, filter $\phi$, $k$, layer $\ell$, search budget $ef$}
\KwOut{Top-$k$ results satisfying $\phi$}
$R\gets\emptyset$\;
Identify cube list $\mathcal{C} = \{c_1, \ldots, c_m\}$ intersecting $\phi$\;
Build adjacency bitmap: $B[c] = 1$ for all $c \in \mathcal{C}$\;
Initialize priority queue $Q$ with entry points in $\mathcal{C}$\;
\While{top result in $Q$ better than $k$-th result in $R$}{
    $p \gets$ pop closest candidate from $Q$\;
    \ForEach{neighbor $u$ of $p$ (intra-cube or cross-cube)}{
        \If{$u$ is cross-cube edge and $B[u.\text{cube}] = 0$}{
            \textbf{continue}\;
        }
        \If{$\phi(\mathbf{s}_u) = 1$}{
            add $u$ to result set $R$\;
        }
        Add $u$ to $Q$ if not visited\;
        keep top $ef$ results in $R$\;
    }
}
\Return{top-$k$ from $R$}
\end{algorithm}

\stitle{Predetermined Cube Search.}
Algorithm~\ref{alg:predetermined-search} presents the predetermined cube search strategy, suitable for simple filters where all intersecting cubes can be identified upfront. Given a query vector $\mathbf{q}$, filter $\phi$, parameter $k$ and search effort $ef$, we first identify the set $\mathcal{C} = \{c_1, c_2, \ldots, c_m\}$ of cubes intersecting $\phi$ at the selected layer (line 2). We construct an adjacency bitmap for efficient neighbor checking: for each cube in $\mathcal{C}$, we mark it as searchable in a bitmap (line 3). We initialize the search by adding the entry points of all cubes in $\mathcal{C}$ to the candidate queue (line 4). During beam search, we follow intra-cube edges normally, but only follow cross-cube edges to cubes in $\mathcal{C}$ (lines 5-9). The filter predicate $\phi$ is applied to each candidate to ensure only qualifying points are returned (line 8). This approach minimizes overhead by pre-computing the search domain and avoiding dynamic cube discovery.

\stitle{On-the-Fly Merged Search.}
Algorithm~\ref{alg:fly-search} presents the on-the-fly merged search strategy, designed for complex filters where relevant cubes are discovered dynamically during search. Given a query vector $\mathbf{q}$, filter $\phi$, an entry cube $c_0$, and search budget $ef$, we initialize a dynamic cube bitmap $B$ with only $c_0$ marked as searchable (lines 1-2). We perform beam search with the same termination condition as predetermined search: the search continues while the top candidate in the priority queue $Q$ is better than the $k$-th result in $R$ (line 3). For each neighbor $n$ explored during beam search, we first check if its cube is marked as searchable in $B$ (line 4). If the neighbor satisfies the filter $\phi$, we mark its cube as searchable by setting $B[n.\text{cube}] = 1$ and add it to the result set $R$ (lines 5-6). We maintain the top $ef$ results in $R$ to control search budget (line 8). This dynamic discovery mechanism allows the search to naturally expand into relevant cubes as qualifying points are encountered, without requiring upfront computation of all intersecting cubes. This approach is particularly effective for complex filter shapes (circles, polygons) where geometric intersection tests are expensive.

\begin{algorithm}[t]
\caption{On-the-Fly Merged Search}\label{alg:fly-search}
\KwIn{Query $\mathbf{q}$, filter $\phi$, $k$, entry cube $c_0$, search budget $ef$}
\KwOut{Top-$k$ results satisfying $\phi$}
$R\gets\emptyset$\;
Initialize dynamic cube bitmap $B$ with $B[c_0] = 1$\;
Initialize priority queue $Q$ with entry point of $c_0$\;
\While{top result in $Q$ better than $k$-th result in $R$}{
    $p \gets$ pop closest candidate from $Q$\;
    \ForEach{neighbor $n$ of $p$ (intra-cube or cross-cube)}{
        \If{$B[n.\text{cube}] = 0$}{
            \textbf{continue}\tcp{Skip non-searchable cubes}
        }
        \If{$\phi(\mathbf{s}_n) = 1$}{
            $B[n.\text{cube}] \gets 1$\tcp{Mark cube as searchable}
            Add $n$ to result set $R$\;
        }
        Add $n$ to $Q$ if not visited\;
        keep top $ef$ results in $R$\;
    }
}
\Return{top-$k$ from $R$}
\end{algorithm}

\stitle{Comparison and Trade-offs.}
Table~\ref{tab:search-comparison} compares the two query processing approaches. Predetermined cube search is optimal for simple filters (axis-aligned bounding boxes) where cube intersection can be computed efficiently. It has lower per-candidate overhead since the search domain is fixed. On-the-fly merged search excels for complex filters (circles, polygons, irregular regions) where geometric intersection tests are expensive or the filter shape is not known upfront. The dynamic discovery adds overhead (checking filter predicate and updating bitmap), but eliminates the cost of pre-computing all intersecting cubes. In practice, we select the strategy based on filter complexity: use a predetermined search for bounding boxes and an on-the-fly search for other shapes.

\begin{table}[t]
\caption{Comparison of Query Processing Approaches}\label{tab:search-comparison}
\small
\centering
\begin{tabular}{l|c|c}
\toprule
Aspect & Predetermined & On-the-Fly \\
\midrule
Filter Type & Simple (boxes) & Complex (circles, polygons) \\
Cube Discovery & Pre-computed & Dynamic \\
Overhead & Lower & Medium \\
Flexibility & Limited & High \\
Use Case & Known  Cubes & Unknown Intersection \\
\bottomrule
\end{tabular}
\end{table}

\subsection{Dynamic Updates}\label{sec:dynamic-updates}

$\CG$ supports dynamic insertions and deletions while maintaining index quality and query performance. We describe the update procedures and analyze their complexity.

\stitle{Point Insertion.}
When a new data point $(o, \mathbf{x}, \mathbf{s})$ arrives, we insert it into the index as follows. First, we compute the cube ID for the point at each layer $\ell \in [0, L-1]$ based on its metadata $\mathbf{s}$ (same computation as in construction). For each layer, we insert the point into the local graph index of its containing cube using the standard graph insertion algorithm, which connects the new point to its $M$ nearest neighbors. We add cross-cube edges for the new point by searching from the entry points of adjacent cubes and connecting to the top $M_{\text{cross}}$ nearest neighbors in each adjacent cube. The insertion complexity is $O(L \cdot m \log N)$, dominated by cross-edge addition over the $2m$ adjacent cubes at each of the $L$ layers ($L$ is a small constant in practice).

\stitle{Point Deletion.}
We adopt a lazy deletion strategy for efficiency. When a point is deleted, we mark it as invalid in all containing cubes across all layers. During query processing, we skip invalidated points when they appear in the candidate queue. Periodically (e.g., when the deletion rate exceeds a threshold), we rebuild affected cube indices to reclaim memory and maintain search efficiency. Lazy deletion has $O(L)$ complexity for marking the affected entries.

\section{Analysis and Extension}\label{sec:extend}

In this section, we provide a theoretical analysis of $\CG$ using the framework of characteristic length, cube length, and elastic factor~\cite{ELI-VLDB-2026-Mingyu}. We analyze space and time complexity under the uniform distribution assumption, and discuss extensions to the $\CG$ framework for user-specific indexes.

\subsection{Analysis}\label{sec:analysis}

\stitle{Uniform Distribution Model.}
We assume the metadata space $\mathcal{S} = [0, S]^m$ is an $m$-dimensional hypercube. Given a dataset $\mathcal{D}$ of $N$ points, we assume the metadata vectors $\mathbf{s}_i$ are independently and uniformly distributed in $\mathcal{S}$. The hierarchical grid has $L$ levels with granularity parameter $g$, where level $\ell$ partitions $\mathcal{S}$ into $g^\ell$ cubes. We denote the cube width (side length) at layer $\ell$ as $w_\ell = S / g^\ell$. In practice, $L$ is set to a small constant (independent of $N$); filters with characteristic length smaller than the deepest cube width $w_L$ are simply routed to the bottom layer rather than triggering deeper subdivision, so $L$ remains $O(1)$ throughout the analysis.

\stitle{Characteristic Length.}
For a given filter $\phi$, we define its \emph{characteristic length} $r$ as follows: for an axis-aligned bounding box, $r = r_{\max}$ is the maximum side length; for a circular or spherical filter, $r$ is the diameter. For rectangular filters, we also define $r_{\min}$ as the minimum side length and the \emph{aspect ratio} $\alpha = r_{\max} / r_{\min} \ge 1$. The characteristic length captures the spatial extent of the filter in the metadata space, while the aspect ratio characterizes the shape of the filter ($\alpha = 1$ for squares, $\alpha > 1$ for non-square rectangles).

\stitle{Elastic Factor.}
We adapt the elastic factor concept from \cite{ELI-VLDB-2026-Mingyu,ESG-arxiv-2025-Mingyu} to characterize query efficiency in \CG{}. The elastic factor measures the overlap between the filtered candidate set and the index searched set.

\begin{definition}[Elastic Factor for CubeGraph]\label{def:elastic-factor}
Given a dataset $\mathcal{D}$, a query $(\mathbf{q}, \phi)$, and the merged graph $\mathcal{G}^*$ at layer $\ell$, the \emph{elastic factor} is defined as:
\begin{equation*}
    e(\mathcal{D}_\phi, \mathcal{G}^*) = \frac{|\mathcal{D}_\phi|}{|\mathcal{G}^*|} = \frac{|\{(o, \mathbf{x}, \mathbf{s}) \in \mathcal{D} \mid \phi(\mathbf{s}) = 1\}|}{|\{(o, \mathbf{x}, \mathbf{s}) \in \mathcal{D} \mid \mathbf{s} \in \bigcup \mathcal{C}\}|}
\end{equation*}
where $\mathcal{C}$ is the set of cubes intersecting $\phi$ at layer $\ell$.
\end{definition}

The elastic factor $e \in (0, 1]$ measures what fraction of the searched points actually satisfy the filter. Under uniform distribution, the elastic factor can be approximated by the volume ratio:
\begin{equation*}
    e \approx \frac{\text{Vol}(\phi)}{\text{Vol}(\bigcup \mathcal{C})}
\end{equation*}
where $\text{Vol}(\cdot)$ denotes the $m$-dimensional volume.

\stitle{Optimal Layer Selection.}
\begin{prop}\label{prop:optimal-layer}
For a filter $\phi$ with characteristic length $r$, choosing layer $\ell^*$ such that $r/2 < w_{\ell^*} \le r$ guarantees that $\phi$ intersects at most $3^m$ cubes; this choice is optimal in the sense that any other layer increases either the cube count or the elastic factor lower bound (see technical report~A.1).
\end{prop}

\begin{lem}\label{lem:elastic-factor-bound}
Under uniform distribution with optimal layer selection (Proposition~\ref{prop:optimal-layer}), where $w_{\ell^*} \approx r$, the elastic factor is lower bounded by a constant that depends on the filter shape. For circular filters in $m$ dimensions, $e \ge \pi^{m/2} / (6^m \cdot \Gamma(m/2 + 1))$, where $\Gamma$ is the gamma function.
\end{lem}

\begin{proof}
The filter with characteristic length $r$ intersects at most $3^m$ cubes (Proposition~\ref{prop:optimal-layer}), each with volume $w_{\ell^*}^m \approx r^m$. The union of cubes therefore has volume at most $3^m \cdot r^m$. For a circular filter with diameter $r$, the volume is $\pi^{m/2} \cdot (r/2)^m / \Gamma(m/2 + 1)$. Therefore, $e \ge \pi^{m/2} \cdot (r/2)^m / (\Gamma(m/2 + 1) \cdot 3^m \cdot r^m) = \pi^{m/2} / (6^m \cdot \Gamma(m/2 + 1))$. For $m=2$, $e \ge \pi/36 \approx 0.087$. The lower bound is a positive constant independent of $N$.
\end{proof}

\begin{exmp}
Consider a 2D circular filter with radius $r/2$ (diameter $r$) in a metadata space of side length $S = 100$. With optimal layer selection where $w_{\ell^*} \approx r$, the filter intersects at most {$3^2 = 9$} cubes. The circle has area $\pi (r/2)^2 = \pi r^2/4$, and the union of 9 cubes has area at most $9 r^2$. The elastic factor is $e \approx \pi r^2/4 / {(9 r^2)} = {\pi/36 \approx 0.087}$. In practice, the elastic factor is often higher because the filter may not span all 9 cubes fully.
\end{exmp}

The elastic factor degrades with aspect ratio: for a rectangular filter with characteristic length $r = r_{\max}$ and $\alpha = r_{\max}/r_{\min}$, the long dimension contributes at most $3$ cubes and each shorter dimension contributes at most $2$ cubes (since $r_{\min} < w_{\ell^*}$), giving a cube union of $3 \cdot 2^{m-1} \cdot r^m$ in the worst case. The filter volume is $r \cdot r_{\min}^{m-1} = r^m / \alpha^{m-1}$, so $e \approx 1 / (3 \cdot 2^{m-1} \cdot \alpha^{m-1})$, which degrades as a polynomial of degree $m-1$ in $\alpha$. See technical report~A.4 for detailed analysis of query performance degradation and technical report~A.5 for an example of high aspect ratio rectangles.

\begin{exmp}
Consider a 2D rectangular filter with sides $100 \times 10$ (aspect ratio $\alpha = 10$) in a metadata space of side length $S = 1000$. With optimal layer selection where $w_{\ell^*} \approx 100$, the rectangle intersects at most $3 \times 2 = 6$ cubes (3 along the long dimension, 2 along the short dimension since $r_{\min} = 10 < w_{\ell^*}$). The rectangle has area $1000$, and the union of $6$ cubes has area approximately $6 \times 100^2 = 60{,}000$. The elastic factor is $e \approx 1000 / 60{,}000 \approx 0.017$, much lower than the $\pi/36 \approx 0.087$ bound for circular filters at the same layer. This demonstrates why high aspect ratio filters lead to poor query performance in \CG{}.
\end{exmp}

\stitle{Space Complexity.}
We analyze the space complexity of \CG{} under uniform distribution.

\begin{thm}\label{thm:space-complexity}
Under uniform metadata distribution, the space complexity of \CG{} is $O(N \cdot m)$, where $N$ is the dataset size and $m$ is the metadata dimensionality. The hierarchy depth $L$, intra-cube max degree $M$, and cross-cube max degree $M_{\text{cross}}$ are treated as constants of the index configuration.
\end{thm}

\begin{proof}
Each data point appears in exactly one cube at each of the $L$ layers. At each layer, the point maintains $O(1)$ intra-cube edges and $O(m)$ cross-cube edges (one set per adjacent cube, of which there are $2m$ in an $m$-dimensional grid). Treating $L$, $M$, and $M_{\text{cross}}$ as constants, each point stores $O(m)$ edges per layer, yielding total space $O(L \cdot N \cdot m) = O(N \cdot m)$.
\end{proof}

\stitle{Construction Time Complexity.}
We analyze the time complexity of index construction.

\begin{thm}\label{thm:construction-time}
The construction time complexity of \CG{} is $O(N \cdot L \cdot m \cdot \log N)$, where $L$ is the number of hierarchy layers ($L$ is a small constant in practice).
\end{thm}

\begin{proof}
Construction consists of two phases: (1) building local graph indices within each cube, and (2) adding cross-cube edges. For phase (1), each of the $N$ points is inserted into the local graph at each of the $L$ layers, with each insertion costing $O(\log N)$ on average for graph-based indices like HNSW, yielding $O(N \cdot L \cdot \log N)$ time. For phase (2), each point at each layer searches in $2m$ adjacent cubes; each search costs $O(\log N)$, yielding $O(N \cdot L \cdot m \cdot \log N)$ time. The total construction time is dominated by phase (2): $O(N \cdot L \cdot m \cdot \log N)$.
\end{proof}

\stitle{Remark.}
Recent studies~\cite{RNSM-Arxiv-2026-Mingyu} show that cross-cube links require only minimal search effort ($ef_c = 30$) compared to full graph construction ($ef_c = 200$). See our technical report for details.

\stitle{Query Time Complexity.}
We now analyze the query time complexity, which depends critically on the relationship between characteristic length $r$, cube width $w_\ell$, and elastic factor $e$.

\begin{thm}\label{thm:query-time}
Given a query $(\mathbf{q}, \phi)$ with characteristic length $r$, selecting the optimal layer $\ell^*$ where $w_{\ell^*} \approx r$ (as in Proposition~\ref{prop:optimal-layer}), if the elastic factor $e \ge c$ for some constant $c \in (0, 1]$, the expected query time to retrieve top-$k$ results is $O(C + k/c)$, where $C$ is the expected cost to locate the top-1 neighbor in the merged graph $\mathcal{G}^*$.
\end{thm}

\stitle{Remark.} The merged graph $\mathcal{G}^*$ exhibits the same recall-throughput trade-off as a monolithic graph index built over the same union of points; that is, $C$ inherits the cost of the underlying graph index (e.g., $C = O(\log N)$ on for HNSW~\cite{Acorn-SIGMOD-2024,HNSW-Merge-SIGMOD-2026-Jianguo-Wang,hnswlib}). This equivalence is validated empirically in previous works~\cite{RNSM-Arxiv-2026-Mingyu,HNSW-Merge-SIGMOD-2026-Jianguo-Wang}, where \CG{}'s recall-Qps curves match or exceed those of monolithic baselines on the same data.

\begin{cor}\label{cor:query-time-independence}
Under optimal layer selection with elastic factor $e \ge c$, the query time of \CG{} is $O(C + k/c)$, which is independent of the dataset size $N$ (apart from $C$'s graph-search dependence) and depends only on the filter geometry (through $r$ and $e$) and the graph structure (through $C$).
\end{cor}

\stitle{Dynamic Update Complexity.}
We briefly analyze the complexity of dynamic updates.

\begin{thm}\label{thm:update-complexity}
Point insertion has time complexity $O(L \cdot m \log N)$, and lazy deletion has time complexity $O(L)$ for marking a point as invalid, where $L$ is the number of hierarchy layers ($L$ is a small constant in practice).
\end{thm}

\begin{proof}
For insertion, the point is inserted into each of the $L$ layers. At each layer, inserting into the local graph costs $O(\log N)$, and establishing cross-cube edges to the $2m$ adjacent cubes costs $O(m \log N)$. Summing over $L$ layers gives total insertion cost $O(L \cdot m \log N)$. For lazy deletion, marking the point as invalid in each layer is $O(1)$ per layer, so total cost is $O(L)$.
\end{proof}

\section{Experiments}\label{sec:exp}
Our experiments evaluate $\CG$ across following dimensions: search efficiency compared to state-of-the-art; query performance under diverse filter shapes and spatial distributions; the effectiveness for scalability.

\subsection{Experimental Setup}\label{sec:exp-setup}

\stitle{Datasets.}
Table~\ref{tab:dataset_details} summarizes the statistics of the datasets used in our experiments.
We use the following standard ANN benchmarks:
\begin{itemize}[leftmargin=*,topsep=2pt,itemsep=1pt]
  \item \textbf{SIFT1M}~\cite{ann-benchmakrs}: 1M 128-dimensional SIFT descriptors with 10K queries. We generate synthetic 2D/3D/4D various distribution attributes in $[0,1]^m$ for spatio-temporal filtering.
  \item \textbf{YFCC}: 1M 512-dimensional CLIP embeddings extracted from Flickr images with real geolocation metadata (latitude/longitude) and timestamp (normalized). We use the first 1M vectors for our experiments.
  \item \textbf{WIKI}: 1.16M 384-dimensional sentence embeddings from the WikiText-103 corpus,
   encoded with the sentence transformer. Each embedding represents a Wikipedia article     
  passage. We attach real 2D geolocation metadata (latitude/longitude) and timestamp (normalized). 
  \item \textbf{MSMARC10M}: 10M 1024-dimensional text embeddings from the MS MARCO passage ranking dataset. We generate synthetic 2D/3D/4D uniform attributes for filtered search evaluation.
  \item \textbf{Deep100M}: 100M 96-dimensional vectors sampled from the Deep1B dataset~\cite{ann-benchmakrs}. We generate 2D/3D uniform spatial attributes to test scalability.
\end{itemize}

\begin{table}[!t] 
\centering 
\caption{The Statistics of Datasets}\vspace{-2ex}
\label{tab:dataset_details} 
\begin{footnotesize}
\begin{tabular}{l|ccccc}
\toprule
\textbf{Dataset} & \textbf{Size} & \textbf{Dim} & \textbf{Query Size} & \textbf{Metadata} \\
\midrule
SIFT1M & 1,000,000 & 128 & 10,000 & 2D/3D/4D Uniform \\
YFCC & 999,894 & 512 & 1,000 & 2D/3D Geo \\
WIKI & 1,170,381 & 384 & 1,000 & 2D/3D Geo \\
MSMARC10M & 10,000,000 & 1024 & 1,000 & 2D/3D/4D Uniform \\
Deep100M & 100,000,000 & 96 & 1,000 & 2D/3D Uniform \\
\bottomrule
\end{tabular}
\end{footnotesize}
\end{table}

\stitle{Baselines.}
We compare against the following state-of-the-art methods:
\begin{itemize}[leftmargin=*,topsep=2pt,itemsep=1pt]
  \item $\POST$: $\HNSW$ with post-filtering, which applies the filter predicate after retrieving candidates from a standard HNSW index.
  \item $\ACORN$-$\gamma$: A method that constructs a dense graph index for the entire dataset. We use $\gamma=12$ as recommended for filtered search.
  \item \textbf{iRangeGraph}~\cite{iRangeGraph-SIGMOD-2025}: A range-filtering ANNS method that organizes graph indices in a segment-tree-like structure over a 1D sorted attribute. iRangeGraph supports 2D metadata filter by allowing out-of-range neighbors during search.
  \item \textbf{Tree-Graph}: A multi-dimensional tree-graph hybrid (representative of the family discussed in Section~\ref{sec:problem}) that organizes per-node graph indices in a KD-tree over the metadata space. A query traverses the tree to identify overlapping leaves, runs a graph search on each, and merges the results.
\end{itemize}

\stitle{Evaluation Metrics.}
We evaluate the performance of our method using the following metrics:
\begin{itemize}[leftmargin=*,topsep=2pt,itemsep=1pt]
  \item \textbf{Recall}: For a given query, let \( R \) be the set of the exact \( k \)-nearest neighbors (ground truth) and \( A \) be the set of \( k \) neighbors returned by the approximate search. Recall is defined as \( {|R \cap A|}/{k} \).
  \item \textbf{Query Per Second (Qps)}: The number of queries processed per second.
\end{itemize}
Note, all metrics are averaged over the entire query set for each dataset and filter configuration.

\stitle{Metadata Distribution.}
For SIFT and MSMARC10M, we generate synthetic metadata attributes uniformly distributed in $[0,1]^m$ for $m=2,3,4$. For YFCC, we use real geolocation metadata (latitude/longitude) and timestamp. For Deep100M, we generate synthetic 2D/3D uniform spatial attributes to test scalability.
\begin{figure}[!t]
    \centering
    \includegraphics[width=0.6\columnwidth]{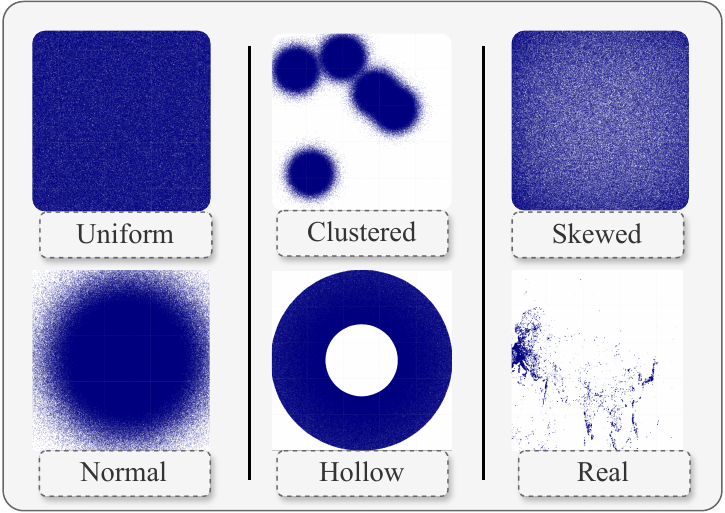}
    \vgap\caption{Distribution of Metadata Attributes Across Datasets.}\vgap
    \label{fig:meta-distribution}
\end{figure}

\stitle{Query Workloads.}
We design query workloads with varying filter shapes and sizes:
\begin{itemize}[leftmargin=*,topsep=2pt,itemsep=1pt]
  \item \textbf{Axis-Aligned Bounding Boxes}: Rectangular filters with varying edge length with a fluctuation of approximately 20\%.
  \item \textbf{Circles}: Circular filters with varying radius.
  \item \textbf{Polygons}: Irregular filters defined by random polygons with 3-5 vertices.
  \item \textbf{Compose}: Complex filters formed by combining basic shapes (e.g., points inside a bounding box but outside a circle).
\end{itemize}      

\stitle{Filter Ratios.}
We vary the filter ratio (the fraction of volume to the metaspace) from 0.01 to 0.10 to evaluate performance under different selectivity levels. If metadata is uniformly distributed, the filter ratio directly corresponds to the expected fraction of points satisfying the filter. For specific or real data, the selectivity will be less than the filter ratio. 
For synthetic filters, we control the filter ratio by adjusting the size of the filter (e.g., side length for bounding boxes, radius for circles).

\input{figure/Exp-1}

\input{figure/Exp-2}

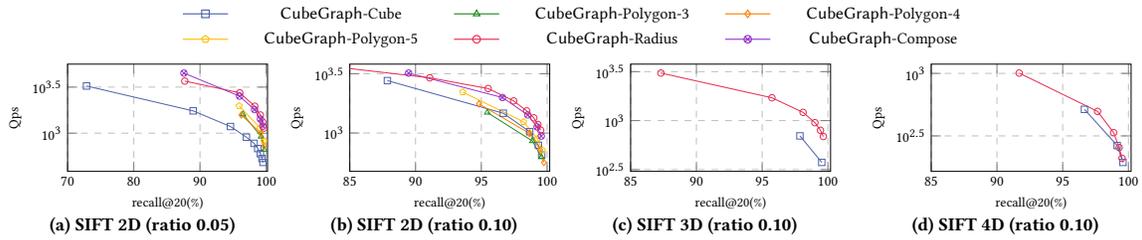
\begin{figure*}[!t]
\centering
\begin{footnotesize}
\begin{tikzpicture}
    \begin{customlegend}[legend columns=3,
        legend entries={$\CG$-Cube,$\CG$-Polygon-3,$\CG$-Polygon-4,$\CG$-Polygon-5,$\CG$-Radius, $\CG$-Compose},
        legend style={at={(0.5,1.15)},anchor=north,draw=none,font=\scriptsize,column sep=0.4cm}]
    \addlegendimage{line width=0.15mm,color=navy,mark=square,mark size=0.5mm}
    \addlegendimage{line width=0.15mm,color=forestgreen,mark=triangle,mark size=0.5mm}
    \addlegendimage{line width=0.15mm,color=orange,mark=diamond,mark size=0.5mm}
    \addlegendimage{line width=0.15mm,color=amber,mark=pentagon,mark size=0.5mm}
    \addlegendimage{line width=0.15mm,color=amaranth,mark=o,mark size=0.5mm}
    \addlegendimage{line width=0.15mm,color=violet,mark=otimes,mark size=0.5mm}
    \end{customlegend}
\end{tikzpicture}

\subfloat[SIFT 2D (ratio 0.05)]{\vgap
\begin{tikzpicture}[scale=0.85]
\begin{axis}[
    height=\columnwidth/2.60,
    width=\columnwidth/1.80,
    xlabel=recall@20(\%),
    ylabel=Qps,
    ymode=log,
    xmin=70,
    xmax=100.2,
    label style={font=\scriptsize},
    tick label style={font=\scriptsize},
    title style={font=\scriptsize},
    ymajorgrids=true,
    xmajorgrids=true,
    grid style=dashed,
]
\addplot[line width=0.15mm,color=navy,mark=square,mark size=0.5mm]
plot coordinates {
    (72.86, 3252.27)
    (88.96, 1745.92)
    (94.60, 1180.46)
    (97.01, 903.60)
    (98.17, 775.26)
    (98.73, 680.41)
    (99.12, 597.86)
    (99.36, 531.44)
    (99.52, 481.10)
};
\addplot[line width=0.15mm,color=forestgreen,mark=triangle,mark size=0.5mm]
plot coordinates {
    (96.49, 1601.97)
    (99.24, 920.62)
    (99.72, 662.67)
};
\addplot[line width=0.15mm,color=orange,mark=diamond,mark size=0.5mm]
plot coordinates {
    (96.25, 1576.95)
    (99.18, 1000.01)
    (99.69, 729.50)
};
\addplot[line width=0.15mm,color=amber,mark=pentagon,mark size=0.5mm]
plot coordinates {
    (95.92, 1978.12)
    (99.08, 1139.69)
    (99.66, 818.46)
};
\addplot[line width=0.15mm,color=amaranth,mark=o,mark size=0.5mm]
plot coordinates {
(87.70, 3677.87)
(96.01, 2741.80)
(98.27, 1962.34)
(99.12, 1574.56)
(99.48, 1320.29)
(99.68, 1136.55)
};
\addplot[line width=0.15mm,color=violet,mark=otimes,mark size=0.5mm]
plot coordinates {
(87.57, 4493.37)
(95.98, 2520.59)
(98.24, 1801.58)
(99.13, 1422.69)
(99.51, 1191.04)
};
\end{axis}
\end{tikzpicture}
}
\subfloat[SIFT 2D (ratio 0.10)]{\vgap
\begin{tikzpicture}[scale=0.85]
\begin{axis}[
    height=\columnwidth/2.60,
    width=\columnwidth/1.80,
    xlabel=recall@20(\%),
    ylabel=Qps,
    ymode=log,
    xmin=85,
    xmax=100.2,
    label style={font=\scriptsize},
    tick label style={font=\scriptsize},
    title style={font=\scriptsize},
    ymajorgrids=true,
    xmajorgrids=true,
    grid style=dashed,
]
\addplot[line width=0.15mm,color=navy,mark=square,mark size=0.5mm]
plot coordinates {
    (87.86, 2767.05)
    (96.67, 1469.00)
    (98.68, 1025.32)
    (99.34, 788.62)
    (99.59, 652.26)
};
\addplot[line width=0.15mm,color=forestgreen,mark=triangle,mark size=0.5mm]
plot coordinates {
(95.46, 1489.54)
(98.88, 856.54)
(99.57, 626.90)
};
\addplot[line width=0.15mm,color=orange,mark=diamond,mark size=0.5mm]
plot coordinates {
(94.82, 1762.35)
(98.65, 993.84)
(99.48, 721.84)
(99.76, 565.17)
};
\addplot[line width=0.15mm,color=amber,mark=pentagon,mark size=0.5mm]
plot coordinates {
(93.59, 2212.77)
(98.22, 1240.81)
(99.27, 895.06)
(99.65, 714.24)
};
\addplot[line width=0.15mm,color=violet,mark=otimes,mark size=0.5mm]
plot coordinates {
(89.47, 3203.94)
(96.61, 1992.31)
(98.53, 1420.86)
(99.24, 1126.56)
(99.56, 941.64)
};
\addplot[line width=0.15mm,color=amaranth,mark=o,mark size=0.5mm]
plot coordinates {
(78.29, 4261.73)
(91.08, 2923.10)
(95.52, 2372.51)
(97.46, 1866.56)
(98.45, 1537.91)
(98.99, 1342.59)
(99.30, 1180.66)
(99.51, 1057.98)
};
\end{axis}
\end{tikzpicture}
}
\subfloat[SIFT 3D (ratio 0.10)]{\vgap
\begin{tikzpicture}[scale=0.85]
\begin{axis}[
    height=\columnwidth/2.60,
    width=\columnwidth/1.80,
    xlabel=recall@20(\%),
    ylabel=Qps,
    ymode=log,
    xmin=85,
    xmax=100.2,
    label style={font=\scriptsize},
    tick label style={font=\scriptsize},
    title style={font=\scriptsize},
    ymajorgrids=true,
    xmajorgrids=true,
    grid style=dashed,
]
\addplot[line width=0.15mm,color=navy,mark=square,mark size=0.5mm]
plot coordinates {
    (97.87, 698.33)
    (99.53, 374.22)
};
\addplot[line width=0.15mm,color=amaranth,mark=o,mark size=0.5mm]
plot coordinates {
(87.30, 3057.50)
(95.74, 1710.16)
(98.11, 1213.15)
(99.01, 952.38)
(99.43, 795.89)
(99.65, 687.73)
};
\end{axis}
\end{tikzpicture}
}\hspace{2mm}
\subfloat[SIFT 4D (ratio 0.10)]{\vgap
\begin{tikzpicture}[scale=0.85]
\begin{axis}[
    height=\columnwidth/2.60,
    width=\columnwidth/1.80,
    xlabel=recall@20(\%),
    ylabel=Qps,
    ymode=log,
    xmin=85,
    xmax=100.2,
    label style={font=\scriptsize},
    tick label style={font=\scriptsize},
    title style={font=\scriptsize},
    ymajorgrids=true,
    xmajorgrids=true,
    grid style=dashed,
]
\addplot[line width=0.15mm,color=navy,mark=square,mark size=0.5mm]
plot coordinates {
    (96.65, 515.69)
    (99.14, 264.88)
    (99.58, 194.08)
};
\addplot[line width=0.15mm,color=amaranth,mark=o,mark size=0.5mm]
plot coordinates {
    (91.69, 1006.01)
    (97.65, 496.68)
    (98.87, 335.84)
    (99.29, 255.17)
    (99.51, 209.63)
};
\end{axis}
\end{tikzpicture}
}

\caption{Comparison of Polygon and Radius filters vs Cube filter on SIFT (recall@20 vs.\ Qps).}
\label{fig:Exp-3}\vspace{-3ex}
\end{footnotesize}
\end{figure*}

\subsection{Experimental Results}

\stitle{Exp-1: Search Efficiency.}
Fig~\ref{fig:Exp-1} compares $\CG$ against $\ACORN$ and $\POST$ across SIFT, MSMARC10M, and YFCC with varying 2D filter ratios (0.01--0.10). On SIFT, $\CG$ achieves up to 5,730 Qps at 92\% recall---$72\times$ and $21\times$ speedup over $\ACORN$ and $\POST$, respectively. On MSMARC10M, $\CG$ sustains 144 Qps at 99\%+ recall while $\POST$ drops to 15 Qps, with $\ACORN$ unable to exceed 88\% recall. YFCC shows the largest gap: $\CG$ delivers 100$\times$ speedup over $\POST$ at comparable recall, while $\ACORN$ saturates at only 24\% recall. Across all datasets, $\CG$ achieves $1$--$2$ orders of magnitude higher throughput than baselines.

\stitle{Exp-2: Multi-Dimensional Filters.}
We evaluate $\CG$ on queries with varying attribute dimensions (2D, 3D, 4D). Fig~\ref{fig:Exp-2} presents the performance on SIFT with a box filter at different filter ratios. With 2D attributes and 10\% filter ratio, $\CG$ achieves 2,767 Qps at 88\% recall@20 and 652 Qps at 99.6\% recall. Increasing dimensionality to 3D provides finer spatial filtering granularity, achieving 1,469 Qps at 97\% recall. At 4D, $\CG$ maintains 515 Qps at 98\% recall, demonstrating that higher dimensions slightly reduce throughput due to increased intersection complexity but still deliver excellent performance. These results confirm that $\CG$ efficiently handles multi-dimensional spatio-temporal filters.

\stitle{Exp-3: Handling Complex Filter.}
Fig~\ref{fig:Exp-3} evaluates $\CG$ with various filter shapes: box, polygon (3/4/5 vertices), radius, and composed filters on SIFT. At 2D ratio of 0.05, Polygon-5 achieves 1,978 Qps at 96\% recall, which is $1.8\times$ higher than Cube. This demonstrates that irregular filter shapes can reduce intersection overhead. Radius filter achieves 1,136 Qps at 99.7\% recall, while the composed filter (Inside Nox but not in Radius) reaches 1,191 Qps at 99.5\% recall. At 2D ratio of 0.10, Radius maintains 1,058 Qps at 99.5\% recall with Cube at 652 Qps. In 3D, Radius achieves 687 Qps at 99.7\% recall versus Cube 374 Qps---$1.8\times$ speedup. In 4D, both shapes perform comparably (209 vs 194 Qps at 99.5\% recall). These results confirm $\CG$ adapts efficiently to diverse filter geometries.

\stitle{Exp-4: Index Time and Space.}
Table~\ref{tab:index-time} reports the index construction time and space usage for $\CG$ and $\POST$ across four datasets ranging from 1M to 100M vectors. $\CG$ incurs moderate construction overhead compared to $\POST$ due to the hierarchical grid partitioning and cross-cube edge establishment. Our hierarchy terminates when cubes contain fewer than 50 nodes, ensuring sufficient points per leaf cube for effective graph navigation; empirically, 6 layers suffice for most datasets. However, this one-time construction cost is amortized over the entire query workload, and the resulting index structure enables dramatically faster query processing as demonstrated in Exp-1. The construction time scales linearly with dataset size, reaching approximately 5 hours for the Deep100M dataset, which is acceptable for offline index building. 
Note that both iRangeGraph and Tree-Graph fail to scale on the larger datasets in our suite: their index construction either exceeds a 24-hour wall-clock budget or terminates with a core dump (out-of-memory) on Deep100M. 

\begin{table}[!t]
\begin{footnotesize}
    \centering
    \caption{Indexing Time (seconds)}\label{tab:index-time}\vspace{-2ex}
    \begin{tabular}{l|ccccc}
        \toprule
        ~                 & \textbf{SIFT} & \textbf{YFCC} & \textbf{WIKI} & \textbf{MSMARC10M} & \textbf{Deep100M} \\
        \midrule
        $\CG$             & 88            & 251           & 397           & 4767               & 16224 \\
        $\POST$           & 47            & 50            & 58            & 647                & 2467  \\
        $\ACORN$-$\gamma$ & 49            & 133           & 180           & 3572               & 10484 \\
        $\KD$             & 153           & 328           & 556           & 12494              & --    \\
        $\IRANGE$         & 494           & 1007          & 3261          & 20192              & --    \\
        \bottomrule
    \end{tabular}
\end{footnotesize}
\end{table}

\begin{table}[!t]
\begin{footnotesize}
  \centering
  \caption{Index Size (MB)}\label{tab:index-size}\vspace{-2ex}
  \begin{tabular}{l|ccccc}
      \toprule
      ~                 & SIFT & WIKI   & YFCC   & MSMARC10M & Deep100M \\
      \midrule
      Dataset           & 489  & 1767   & 1956   & 39062     & 37004 \\
      $\CG$             & 1038 & 1215   & 1038   & 10376     & 102996 \\
      $\POST$           & 172  & 156    & 172    & 1716      & 17166   \\
      $\ACORN$-$\gamma$ & 363  & 2000   & 360    & 3638      & 36724  \\
      $\KD$             & 1541 & 1961   & 1540   & 20880     & -- \\
      $\IRANGE$         & 888  & 1228   & 873    & 11934     & -- \\
      \bottomrule
  \end{tabular}
\end{footnotesize}
\end{table}

\stitle{Exp-5: Dynamic Update.} Fig.~\ref{fig:dynamic-update} evaluates the efficiency of incremental updates against rebuild-from-scratch on SIFT (a) and YFCC (b), each with 100K initial vectors and 32 threads. Across update ratios 10\%--50\%, incremental update is $2.6\times$--$8.5\times$ faster on SIFT (0.7--3.5\,s vs.\ 6.2--9.1\,s) and $2.7\times$--$8.3\times$ faster on YFCC (2.0--9.5\,s vs.\ 16.3--25.7\,s). The speedup is largest at small update ratios because rebuild-from-scratch reprocesses the entire index regardless of how few points actually changed, whereas incremental update touches only the affected cubes. This validates that the hierarchical grid structure absorbs dynamic changes without index reconstruction.

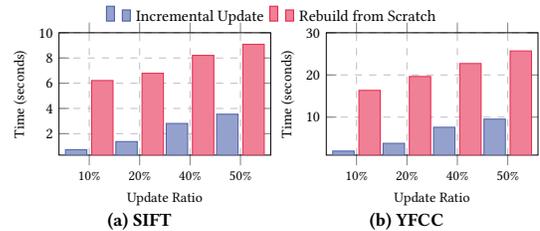
\begin{figure}[!t]
\centering
\begin{footnotesize}
\centerline{\pgfplotslegendfromname{exp5-legend}}\vgap
\subfloat[SIFT]{\vgap
\begin{tikzpicture}[scale=0.90]
\begin{axis}[
    ybar,
    bar width=9pt,
    height=\columnwidth/2.5,
    width=\columnwidth/1.8,
    xlabel={Update Ratio},
    ylabel={Time (seconds)},
    ymin=0.3,
    ymax=10,
    xtick=data,
    xticklabels={10\%,20\%,40\%,50\%},
    label style={font=\scriptsize},
    tick label style={font=\scriptsize},
    title style={font=\scriptsize},
    ymajorgrids=true,
    xmajorgrids=true,
    grid style=dashed,
    legend to name=exp5-legend,
    legend columns=2,
    legend style={
        draw=none,
        font=\scriptsize,
    },
    enlarge x limits=0.2,
]
\addplot[fill=navy!60,draw=navy,line width=0.15mm]
coordinates { (1,0.732) (2,1.374) (3,2.810) (4,3.542) };
\addplot[fill=amaranth!60,draw=amaranth,line width=0.15mm]
coordinates { (1,6.213) (2,6.786) (3,8.208) (4,9.086) };
\legend{Incremental Update, Rebuild from Scratch}
\end{axis}
\end{tikzpicture}
}
\subfloat[YFCC]{\vgap
\begin{tikzpicture}[scale=0.90]
\begin{axis}[
    ybar,
    bar width=9pt,
    height=\columnwidth/2.5,
    width=\columnwidth/1.8,
    xlabel={Update Ratio},
    ylabel={Time (seconds)},
    ymin=1,
    ymax=30,
    xtick=data,
    xticklabels={10\%,20\%,40\%,50\%},
    label style={font=\scriptsize},
    tick label style={font=\scriptsize},
    title style={font=\scriptsize},
    ymajorgrids=true,
    xmajorgrids=true,
    grid style=dashed,
    enlarge x limits=0.2,
]
\addplot[fill=navy!60,draw=navy,line width=0.15mm]
coordinates { (1,1.965) (2,3.796) (3,7.601) (4,9.530) };
\addplot[fill=amaranth!60,draw=amaranth,line width=0.15mm]
coordinates { (1,16.345) (2,19.637) (3,22.708) (4,25.695) };
\end{axis}
\end{tikzpicture}
}
\vgap\caption{Dynamic update time comparison on (a) SIFT and (b) YFCC. Both datasets use 100K initial vectors with 32 threads. Incremental update is consistently faster than full rebuild across all update ratios.}\vgap
\label{fig:dynamic-update}
\end{footnotesize}
\end{figure}

\stitle{Exp-6: Impact of Merge Number.}
We analyze how the number of merge indices affects query performance. Fig.~\ref{fig:recall-qps-merge-num} compares query performance for $4$, $16$, $64$, and $128$ cubes merged on the SIFT dataset. Merge-4 achieves the best performance, reaching over 99\% recall@20 at 478 Qps. As the merge count increases, both recall and throughput degrade significantly. Merge-128 attains only 1/10 search efficiency of Merge-4, demonstrating that excessive graph merging fragments the proximity structure and impairs navigation. These results validate our theoretical analysis that bounded merge counts (proportional to the filter's characteristic length) are essential for maintaining search efficiency.

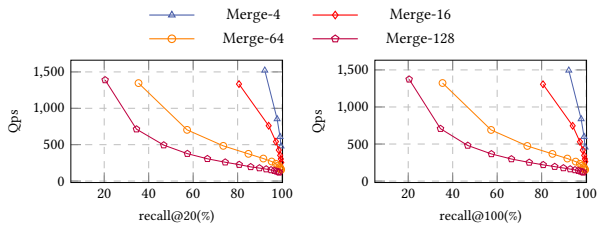
\begin{figure}[!t]
\centering
\begin{footnotesize}
\begin{tikzpicture}
    \begin{customlegend}[legend columns=2,
        legend entries={Merge-4,Merge-16,Merge-64,Merge-128},
        legend style={at={(0.5,1.15)},anchor=north,draw=none,font=\scriptsize,column sep=0.3cm}]
    \addlegendimage{line width=0.15mm,color=navy,mark=triangle,mark size=0.5mm}
    \addlegendimage{line width=0.15mm,color=red,mark=diamond,mark size=0.5mm}
    \addlegendimage{line width=0.15mm,color=orange,mark=o,mark size=0.5mm}
    \addlegendimage{line width=0.15mm,color=purple,mark=pentagon,mark size=0.5mm}
    \end{customlegend}
\end{tikzpicture}

\begin{tikzpicture}[scale=0.9]
\begin{axis}[
height=\columnwidth/2.50,
width=\columnwidth/1.80,
xlabel=recall@20(\%),
ylabel=Qps,
xmin=5,
xmax=100.2,
label style={font=\scriptsize},
tick label style={font=\scriptsize},
title style={font=\scriptsize},
title style={yshift=-2.5mm},
ymajorgrids=true,
xmajorgrids=true,
grid style=dashed,
]
\addplot[line width=0.15mm,color=navy,mark=triangle,mark size=0.5mm]
plot coordinates {
(92.21, 1521.08)
(97.76, 852.62)
(99.07, 608.35)
(99.54, 478.69)
};
\addplot[line width=0.15mm,color=red,mark=diamond,mark size=0.5mm]
plot coordinates {
(80.61, 1333.24)
(93.96, 760.15)
(97.23, 541.80)
(98.51, 425.69)
(99.11, 353.06)
(99.44, 303.50)
(99.63, 255.88)
};
\addplot[line width=0.15mm,color=orange,mark=o,mark size=0.5mm]
plot coordinates {
(35.43, 1346.21)
(57.27, 703.55)
(73.55, 483.69)
(84.85, 373.85)
(91.56, 309.23)
(95.30, 268.42)
(97.17, 237.92)
(98.25, 213.58)
(98.85, 194.19)
(99.19, 178.33)
(99.42, 164.91)
(99.55, 153.70)
};
\addplot[line width=0.15mm,color=purple,mark=pentagon,mark size=0.5mm]
plot coordinates {
(20.42, 1388.88)
(34.51, 713.16)
(46.77, 493.97)
(57.40, 376.40)
(66.41, 306.04)
(74.41, 258.65)
(80.71, 225.46)
(85.88, 194.81)
(89.79, 178.84)
(92.84, 163.98)
(95.02, 152.27)
(96.52, 142.34)
(97.55, 133.78)
(98.27, 126.33)
(98.74, 119.68)
};

\end{axis}
\end{tikzpicture}\hspace{0.5mm}
\begin{tikzpicture}[scale=0.9]
\begin{axis}[
height=\columnwidth/2.50,
width=\columnwidth/1.80,
xlabel=recall@100(\%),
ylabel=Qps,
xmin=5,
xmax=100.2,
label style={font=\scriptsize},
tick label style={font=\scriptsize},
title style={font=\scriptsize},
title style={yshift=-2.5mm},
ymajorgrids=true,
xmajorgrids=true,
grid style=dashed,
]
\addplot[line width=0.15mm,color=navy,mark=triangle,mark size=0.5mm]
plot coordinates {
(92.21, 1493.63)
(97.76, 837.57)
(99.07, 598.22)
(99.54, 456.37)
};
\addplot[line width=0.15mm,color=red,mark=diamond,mark size=0.5mm]
plot coordinates {
(80.61, 1306.53)
(93.96, 746.02)
(97.23, 529.15)
(98.51, 417.86)
(99.11, 346.52)
(99.44, 298.30)
(99.63, 262.74)
};
\addplot[line width=0.15mm,color=orange,mark=o,mark size=0.5mm]
plot coordinates {
(35.43, 1323.24)
(57.27, 690.24)
(73.55, 473.86)
(84.85, 366.72)
(91.56, 304.05)
(95.30, 262.85)
(97.17, 233.00)
(98.25, 209.14)
(98.85, 190.15)
(99.19, 174.58)
(99.42, 161.50)
(99.55, 150.47)
};
\addplot[line width=0.15mm,color=purple,mark=pentagon,mark size=0.5mm]
plot coordinates {
(20.42, 1373.39)
(34.51, 707.20)
(46.77, 480.44)
(57.40, 366.48)
(66.41, 297.98)
(74.41, 251.45)
(80.71, 219.21)
(85.88, 194.78)
(89.79, 176.75)
(92.84, 162.09)
(95.02, 150.13)
(96.52, 140.30)
(97.55, 132.01)
(98.27, 124.51)
(98.74, 118.01)
};

\end{axis}
\end{tikzpicture}\hspace{0.5mm}

\vgap\caption{Recall vs Qps comparison on SIFT dataset ($k=20$: left, $k=100$: right).
We compare the adjacent cube merge with different merge indices count on the SIFT dataset with 2D metadata.}\vgap
\label{fig:recall-qps-merge-num}
\end{footnotesize}
\end{figure}

\stitle{Exp-7: Scalability.}
We evaluate the scalability of $\CG$ on the Deep100M dataset containing 100M 96-dimensional vectors. Table~\ref{tab:index-time} shows that $\CG$ constructs the index in approximately 5 hours (18,508 seconds), demonstrating practical scalability to hundred-million-scale datasets. Fig~\ref{fig:deep100m-scalability} presents the recall@20 vs.\ Qps performance. Despite the massive dataset size, $\CG$ achieves 99.5\% recall at 250 Qps and maintains 99.5\% recall at 255 Qps, demonstrating that the hierarchical grid structure effectively bounds the search space regardless of dataset cardinality. This confirms $\CG$'s suitability for large-scale production deployments.

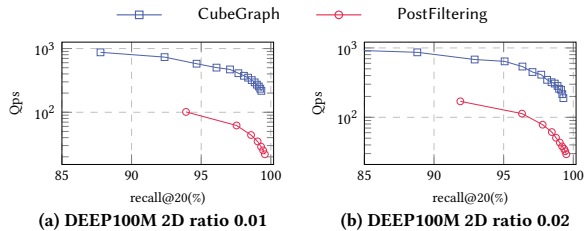
\begin{figure}[!t]
\centering
\begin{footnotesize}
\begin{tikzpicture}
    \begin{customlegend}[legend columns=2,
        legend entries={$\CG$,$\POST$},
        legend style={at={(0.5,1.15)},anchor=north,draw=none,font=\scriptsize,column sep=0.4cm}]
    \addlegendimage{line width=0.15mm,color=navy,mark=square,mark size=0.5mm}
    \addlegendimage{line width=0.15mm,color=amaranth,mark=o,mark size=0.5mm}
    \end{customlegend}
\end{tikzpicture}
\subfloat[DEEP100M 2D ratio 0.01]{\vgap
\begin{tikzpicture}[scale=0.9]
\begin{axis}[
height=\columnwidth/2.50,
width=\columnwidth/1.80,
xlabel=recall@20(\%),
ylabel=Qps,
ymode=log,
xmin=85,
xmax=100.2,
label style={font=\scriptsize},
tick label style={font=\scriptsize},
title style={font=\scriptsize},
ymajorgrids=true,
xmajorgrids=true,
grid style=dashed,
]
\addplot[line width=0.15mm,color=navy,mark=square,mark size=0.5mm]
plot coordinates {
(87.77, 872.56)
(92.37, 734.19)
(94.67, 578.03)
(96.10, 502.17)
(97.07, 470.04)
(97.67, 410.41)
(98.07, 373.71)
(98.37, 349.21)
(98.63, 319.90)
(98.84, 294.77)
(99.00, 271.34)
(99.14, 253.99)
(99.24, 234.07)
(99.32, 216.80)
};
\addplot[line width=0.15mm,color=amaranth,mark=o,mark size=0.5mm]
plot coordinates {
(93.91, 101.30)
(97.54, 62.03)
(98.58, 43.95)
(99.05, 34.89)
(99.30, 28.86)
(99.45, 25.23)
(99.57, 21.92)
};
\end{axis}
\end{tikzpicture}\hspace{2mm}
}
\subfloat[DEEP100M 2D ratio 0.02]{\vgap
\begin{tikzpicture}[scale=0.9]
\begin{axis}[
height=\columnwidth/2.50,
width=\columnwidth/1.80,
xlabel=recall@20(\%),
ylabel=Qps,
ymode=log,
xmin=85,
xmax=100.2,
label style={font=\scriptsize},
tick label style={font=\scriptsize},
title style={font=\scriptsize},
ymajorgrids=true,
xmajorgrids=true,
grid style=dashed,
]
\addplot[line width=0.15mm,color=navy,mark=square,mark size=0.5mm]
plot coordinates {
(78.46, 1011.70)
(88.80, 864.36)
(92.92, 681.00)
(95.07, 640.24)
(96.35, 538.96)
(97.07, 449.37)
(97.70, 411.39)
(98.12, 347.03)
(98.44, 320.70)
(98.67, 306.35)
(98.87, 283.45)
(98.98, 257.50)
(99.12, 249.66)
(99.21, 215.57)
(99.29, 189.10)
};
\addplot[line width=0.15mm,color=amaranth,mark=o,mark size=0.5mm]
plot coordinates {
(91.89, 169.91)
(96.32, 113.05)
(97.80, 78.28)
(98.46, 61.28)
(98.76, 50.75)
(99.03, 43.28)
(99.20, 38.05)
(99.34, 35.11)
(99.42, 32.47)
(99.50, 29.33)
};

\end{axis}
\end{tikzpicture}\hspace{2mm}
}
\vgap\caption{Scalability on Deep100M dataset (100M vectors).
Despite the massive dataset scale, \CG~maintains high search efficiency with Box Filter.}
\label{fig:deep100m-scalability}\vspace{-4ex}
\end{footnotesize}
\end{figure}

\stitle{Exp-8: Various Metadata Distributions.}
Fig.~\ref{fig:Exp-8} shows $\CG$ under five metadata distributions on SIFT 2D: Uniform, Normal, Clustered, Skewed, and Hollow. At a ratio of 0.05, Skewed achieves 1,207 Qps at 99.6\% recall, which is $2.5\times$ higher than Uniform distribution because concentrated data reduces the effective search space. Clustered data yields similar to Uniform (494 Qps). At a ratio of 0.10, Skewed maintains $1.4\times$ speedup, while Hollow performs comparably (847 Qps). Notably, Clustered degrades significantly at a higher ratio (313 Qps, half of Uniform), as dense clusters increase intra-cube competition. These results demonstrate $\CG$ adapts to diverse real-world data distributions, often benefiting from skewed and normal data while being robust to clustered distributions. These experiments also confirm that $\CG$'s stable performance is not solely dependent on uniform distribution, making it suitable for a wide range of applications with varying metadata characteristics.

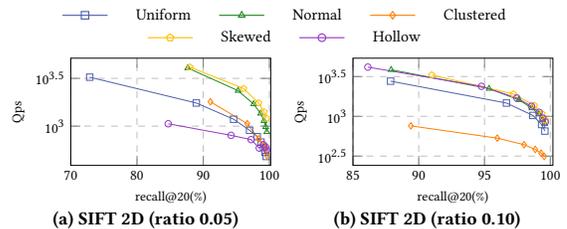
\begin{figure}[!t]
\centering
\begin{footnotesize}
\begin{tikzpicture}\vgap
    \begin{customlegend}[legend columns=3,
        legend entries={Uniform,Normal,Clustered},
        legend style={at={(0.5,1.15)},anchor=north,draw=none,font=\scriptsize,column sep=0.3cm}]
    \addlegendimage{line width=0.15mm,color=navy,mark=square,mark size=0.5mm}
    \addlegendimage{line width=0.15mm,color=forestgreen,mark=triangle,mark size=0.5mm}
    \addlegendimage{line width=0.15mm,color=orange,mark=diamond,mark size=0.5mm}
    \end{customlegend}
\end{tikzpicture}\vgap

\begin{tikzpicture}
    \begin{customlegend}[legend columns=2,
        legend entries={Skewed,Hollow},
        legend style={at={(0.5,1.15)},anchor=north,draw=none,font=\scriptsize,column sep=0.3cm}]
    \addlegendimage{line width=0.15mm,color=amber,mark=pentagon,mark size=0.5mm}
    \addlegendimage{line width=0.15mm,color=violet,mark=o,mark size=0.5mm}
    \end{customlegend}
\end{tikzpicture}

\subfloat[SIFT 2D (ratio 0.05)]{\vgap
\begin{tikzpicture}[scale=0.85]
\begin{axis}[
    height=\columnwidth/2.60,
    width=\columnwidth/1.80,
    xlabel=recall@20(\%),
    ylabel=Qps,
    ymode=log,
    xmin=70,
    xmax=100.2,
    label style={font=\scriptsize},
    tick label style={font=\scriptsize},
    title style={font=\scriptsize},
    ymajorgrids=true,
    xmajorgrids=true,
    grid style=dashed,
]
\addplot[line width=0.15mm,color=navy,mark=square,mark size=0.5mm]
plot coordinates {
    (72.86, 3252.27)
    (88.96, 1745.92)
    (94.60, 1180.46)
    (97.01, 903.60)
    (98.17, 775.26)
    (98.73, 680.41)
    (99.12, 597.86)
    (99.36, 531.44)
    (99.52, 481.10)
};
\addplot[line width=0.15mm,color=forestgreen,mark=triangle,mark size=0.5mm]
plot coordinates {
    (87.67, 4059.25)
    (95.27, 2351.14)
    (97.62, 1702.73)
    (98.68, 1358.27)
    (99.18, 1136.67)
    (99.45, 981.87)
    (99.61, 871.08)
};
\addplot[line width=0.15mm,color=orange,mark=diamond,mark size=0.5mm]
plot coordinates {
    (91.10, 1797.12)
    (96.66, 1060.11)
    (98.39, 737.34)
    (99.08, 630.81)
    (99.43, 586.41)
    (99.63, 493.71)
};
\addplot[line width=0.15mm,color=amber,mark=pentagon,mark size=0.5mm]
plot coordinates {
    (88.00, 4146.75)
    (96.11, 2471.78)
    (98.37, 1747.62)
    (99.20, 1409.55)
    (99.57, 1206.62)
};
\addplot[line width=0.15mm,color=violet,mark=o,mark size=0.5mm]
plot coordinates {
    (84.72, 1058.60)
    (94.21, 800.46)
    (97.24, 721.99)
    (98.47, 588.80)
    (99.09, 639.45)
    (99.41, 605.51)
    (99.61, 562.91)
};
\end{axis}
\end{tikzpicture}
}
\subfloat[SIFT 2D (ratio 0.10)]{\vgap
\begin{tikzpicture}[scale=0.85]
\begin{axis}[
    height=\columnwidth/2.60,
    width=\columnwidth/1.80,
    xlabel=recall@20(\%),
    ylabel=Qps,
    ymode=log,
    xmin=85,
    xmax=100.2,
    label style={font=\scriptsize},
    tick label style={font=\scriptsize},
    title style={font=\scriptsize},
    ymajorgrids=true,
    xmajorgrids=true,
    grid style=dashed,
]
\addplot[line width=0.15mm,color=navy,mark=square,mark size=0.5mm]
plot coordinates {
    (87.86, 2767.05)
    (96.67, 1469.00)
    (98.68, 1025.32)
    (99.34, 788.62)
    (99.59, 652.26)
};
\addplot[line width=0.15mm,color=forestgreen,mark=triangle,mark size=0.5mm]
plot coordinates {
    (87.90, 3853.44)
    (95.35, 2229.90)
    (97.62, 1609.27)
    (98.60, 1269.53)
    (99.13, 1061.29)
    (99.41, 912.75)
    (99.58, 807.94)
};
\addplot[line width=0.15mm,color=orange,mark=diamond,mark size=0.5mm]
plot coordinates {
    (89.40, 758.12)
    (95.97, 531.81)
    (98.01, 439.84)
    (98.86, 383.10)
    (99.30, 343.91)
    (99.54, 312.65)
};
\addplot[line width=0.15mm,color=amber,mark=pentagon,mark size=0.5mm]
plot coordinates {
    (90.99, 3306.55)
    (97.19, 1911.03)
    (98.79, 1370.40)
    (99.38, 1074.41)
    (99.64, 896.69)
};
\addplot[line width=0.15mm,color=violet,mark=o,mark size=0.5mm]
plot coordinates {
    (86.13, 4156.61)
    (94.77, 2382.76)
    (97.46, 1699.80)
    (98.57, 1339.18)
    (99.15, 1124.05)
    (99.43, 949.26)
    (99.62, 847.02)
};
\end{axis}
\end{tikzpicture}
}

\vgap\caption{Impact of metadata distribution on search efficiency.}\vgap
\label{fig:Exp-8}\vspace{-3ex}
\end{footnotesize}
\end{figure}

\section{Related Work}\label{sec:related}

\stitle{Label Filtered Vector Search Methods.}
The utilization of cross-graph indices to handle filtered queries was initially explored in the context of label filtering~\cite{ELI-VLDB-2026-Mingyu,SIEVE-VLDB-2026,Vbase-OSDI-2023,NHQ-NIPS-2022-mengzhao-wang,Filtered-diskann-WWW-2023,Beyond-Vector-Search-Jiadong-Xie-SIGMOD-2025,Tag-Filter-ICDE-2025-Dong-Deng}, as demonstrated by UNG~\cite{UNG-SIGMOD-2025}; however, the performance dynamics of these merged indices were not fully analyzed. UNG leverages the inclusion relationships among labels and dynamically activates cross-graph edges, ensuring that vectors satisfying the label constraints form a navigable graph index. Building upon this, UniFilter~\cite{Beyond-Vector-Search-Jiadong-Xie-SIGMOD-2025} extends the UNG approach to an automaton-based framework. By designing a navigation graph on top of the base index, UniFilter achieves a plug-and-play capability without altering or disrupting the original graph structure. 
In contrast to these methods, \(\CG\) addresses a fundamentally distinct set of problem scenarios. While UNG requires the vectors within the merged index to exactly match the query filter, \(\CG\) adopts a more flexible approach, enabling it to process significantly more complex spatio-temporal queries. Furthermore, the hierarchical structure of \(\CG\) ensures robust search performance across queries of varying granularities. These characteristics make \(\CG\) highly suitable for integration into modern spatio-temporal Retrieval-Augmented Generation systems.

\stitle{Numeric Filtered Vector Search Methods.}
Numeric filtering, which involves selecting vectors based on attributes like prices or timestamps~\cite{DIGRA-SIGMOD-2025-sibo-menxu-cuhk,RangePQ-SIGMOD-2025-fangyuan-sibo,SeRF-SIGMOD-2024,ESG-arxiv-2025-Mingyu,UNIFY-VLDB-2024-liang,Timestamp-ICDE-2025-Yongxin-Tong,Dynamic-RFANNS-VLDB-2025-Dong-Deng,Hi-PNG-KDD-2025-Ming-Yang}. Previous work SeRF~\cite{SeRF-SIGMOD-2024} use compress to reduce the $O(N^2)$ space for all possible range filters, but it is limited to one-dimensional attributes and does not support complex filter shapes. Segment-Tree-based methods~\cite{Window-Filter-ICML-2024,iRangeGraph-SIGMOD-2025,Wow-Range-SIGMOD-2025} build segment-tree-like tree features on top of the base graph index, but they are limited to one-dimensional attributes and do not support complex filter shapes. Hi-PNG~\cite{Hi-PNG-KDD-2025-Ming-Yang} use the similar hierarchical structure of $\CG$ but studies interval-filtering ANNS (IF-ANNS), where both base and query vectors are associated with numerical intervals. 
In contrast, \(\CG\) is designed to handle multi-dimensional spatio-temporal filters with complex geometries, making it more versatile for a wider range of applications. 


\section{Conclusion}\label{sec:conclusion}
In this paper, we presented $\CG$, a highly efficient hierarchical grid index that addresses the core challenges of filtered approximate nearest neighbor search through dynamic graph stitching. By combining a multi-level grid structure with lightweight cross-cube edges, $\CG$ successfully overcomes the trade-off between bounding query intersections and preserving global routing. This architectural advantage translates to significant speedups over existing state-of-the-art baselines while strictly maintaining high recall. Future work includes exploring adaptive grid partitioning strategies to further optimize performance for skewed data distributions and extending the framework to support more complex filter types, such as non-convex shapes or learned filters based on query intent.

\balance

\bibliographystyle{ACM-Reference-Format}
\bibliography{sample-base}

\clearpage

\section*{Appendix}

\subsection*{A.1 Proof of Proposition~\ref{prop:optimal-layer} (Optimal Layer Selection)}

We establish the cube-count bound and then justify the optimality of $w_{\ell^*} \approx r$ against alternative layer choices.

\begin{proof}
\noindent\textbf{Part 1 (Bounding-box decomposition).} Since $\phi$ has characteristic length $r$, its axis-aligned bounding box $B = [b_1^-, b_1^+] \times \cdots \times [b_m^-, b_m^+]$ has side lengths $a_i = b_i^+ - b_i^- \le r$ for every dimension $i$. Any cube intersecting $\phi$ also intersects $B$, so it suffices to bound the cubes intersecting $B$.

\noindent\textbf{Part 2 (Per-dimension cube count).} At layer $\ell^*$ with $r/2 < w_{\ell^*} \le r$, a 1D segment of length $a_i \le r$ touches at most $\lceil a_i / w_{\ell^*} \rceil + 1$ cubes. Since $a_i / w_{\ell^*} < r / (r/2) = 2$, we have $\lceil a_i / w_{\ell^*} \rceil \le 2$, so each dimension touches at most $3$ cubes.

\noindent\textbf{Part 3 (Total cube count).} By independence across dimensions, the bounding box $B$ intersects at most $\prod_{i=1}^{m} n_i \le 3^m$ cubes, establishing the bound stated in Proposition~\ref{prop:optimal-layer}.

\noindent\textbf{Part 4 (Optimality of $w_{\ell^*} \approx r$).} We show that any other layer choice degrades either the cube count or the elastic factor.
\begin{itemize}[leftmargin=*,topsep=2pt,itemsep=1pt]
    \item \emph{Case 1 ($w_\ell \gg r$, cubes too large).} The filter intersects $O(1)$ cubes, so the cube-count bound holds trivially. However, each cube contains $O(N \cdot w_\ell^m / S^m)$ points under uniform distribution, while the filter covers only $O(r^m)$ volume, so $e \approx (r/w_\ell)^m \ll 1$. The search inflates beyond the filter's relevant region.
    \item \emph{Case 2 ($w_\ell \ll r$, cubes too small).} The filter intersects $O((r/w_\ell)^m)$ cubes, which grows unboundedly as $r/w_\ell$ increases. The merged graph becomes fragmented: more cross-cube edges to traverse, weaker small-world connectivity, and less effective beam-search pruning. This degradation is empirically validated in Exp-6 (Section~\ref{sec:exp}), where Merge-128 achieves only $\sim$1/10 the throughput of Merge-4 at the same recall.
    \item \emph{Case 3 (high-aspect-ratio rectangles, $\alpha = r_{\max}/r_{\min} \gg 1$).} With $w_{\ell^*} \approx r = r_{\max}$, the cube-count bound $3^m$ still holds (the long dimension contributes at most $3$ cubes; each shorter dimension contributes at most $2$ cubes since $r_{\min} < w_{\ell^*}$). However, the filter volume $r \cdot r_{\min}^{m-1} = r^m / \alpha^{m-1}$ is much smaller than the cubes' union volume $\Theta(r^m)$, so $e = \Theta(1/\alpha^{m-1})$ and query performance degrades polynomially in $\alpha$.
\end{itemize}

Therefore, $w_{\ell^*}$ with $r/2 < w_{\ell^*} \le r$ is the layer choice that yields the bounded cube count $3^m$ together with a constant elastic-factor lower bound (Lemma~\ref{lem:elastic-factor-bound}).
\end{proof}

\stitle{Limitations.} For high-aspect-ratio filters ($\alpha \gg 1$), Proposition~\ref{prop:optimal-layer} continues to bound the cube count by $3^m$, but the elastic factor decays as $\Theta(1/\alpha^{m-1})$ (Case 3 above), so query latency grows polynomially with $\alpha$. \CG{} does not split filters or escalate layers as a recovery mechanism; this performance loss for high-$\alpha$ workloads is a known limitation of the framework, quantified empirically in technical report~A.5.

\subsection*{A.2 Proof of Lemma~\ref{lem:elastic-factor-bound} (Elastic Factor Bound)}

The filter with characteristic length $r$ intersects at most $3^m$ cubes (Proposition~\ref{prop:optimal-layer}), each with volume $w_{\ell^*}^m \approx r^m$. The union of cubes has volume at most $3^m \cdot r^m$. For a circular filter with diameter $r$, the volume is $\pi^{m/2} \cdot (r/2)^m / \Gamma(m/2 + 1)$. Therefore, $e \ge \pi^{m/2} \cdot (r/2)^m / (\Gamma(m/2 + 1) \cdot 3^m \cdot r^m) = \pi^{m/2} / (6^m \cdot \Gamma(m/2 + 1))$. For $m=2$, $e \ge \pi/36 \approx 0.087$. The lower bound is a positive constant independent of $N$.

\subsection*{A.3 Proof of Theorem~\ref{thm:query-time} (Query Time Complexity)}

The query processing consists of three phases:

\noindent\textbf{Phase 1: Layer selection.} We select the optimal layer $\ell^*$ by binary search over $L$ layers, comparing $w_\ell$ with the characteristic length $r$. This costs $O(\log L)$ time.

\noindent\textbf{Phase 2: Cube identification.} At the optimal layer $\ell^*$, we identify all cubes intersecting the filter $\phi$. By Proposition~\ref{prop:optimal-layer}, there are at most $3^m$ such cubes, a constant independent of dataset size $N$ for fixed dimensionality $m$.

\noindent\textbf{Phase 3: Graph search.} We perform beam search on the merged graph $\mathcal{G}^*$ formed by the $O(1)$ intersecting cubes. Locating the top-1 neighbor costs $O(C)$, where $C$ depends on the graph structure and search parameters. For each additional result, we visit amortized $O(1/c)$ candidates because at least a fraction $c$ of visited neighbors satisfy the filter (due to the elastic factor bound $e \ge c$). Retrieving $k-1$ additional results costs $O(k/c)$.

The total query time is $O(\log L + C + k/c) = O(C + k/c)$ since $C$ dominates for typical values of $L$ and $k$.

\subsection*{A.4 Detailed Discussion: Query Performance Degradation}

The $O(1)$ bound on the number of merged cubes (Proposition~\ref{prop:optimal-layer}) is critical for achieving the $O(C + k/c)$ query time. When $w_\ell \ll r$, the number of merged cubes grows as $O((r/w_\ell)^m)$, which causes significant performance degradation:

\begin{itemize}[leftmargin=*,topsep=2pt,itemsep=1pt]
\item \textbf{More cross-cube edges to traverse:} Each additional cube introduces $O(1)$ cross-cube edges per boundary node, increasing the search space.
\item \textbf{Reduced graph navigability:} Merging many small cube graphs creates a fragmented structure where the small-world property of graph indices degrades.
\item \textbf{Beam search inefficiency:} With a larger, more fragmented search space, beam search becomes less effective at pruning irrelevant candidates.
\end{itemize}

Our experiments (Section~\ref{sec:exp}) validate this analysis by showing that query latency increases significantly as the number of merged cubes grows beyond $3^m$. By selecting the optimal layer where $w_{\ell^*} \approx r$, \CG{} maintains a small constant number of merged cubes, preserving high graph search performance. In $m=2$ this corresponds to at most $9$ merged cubes; growing the merged-graph count to Merge-128 yields roughly $1/10$ the throughput at the same recall (Exp-6), motivating the $w_{\ell^*} \approx r$ rule.

\noindent\textbf{Impact of aspect ratio.} For rectangular filters with high aspect ratio $\alpha = r_{\max} / r_{\min}$, the cube count remains bounded by $3^m$, so graph fragmentation is not the bottleneck. Instead, the elastic factor degrades as $e = \Theta(1/\alpha^{m-1})$, since the filter volume $r \cdot r_{\min}^{m-1} = r^m / \alpha^{m-1}$ shrinks while the cubes' union volume stays $\Theta(r^m)$. Query time therefore grows linearly with $\alpha$ in 2D, since $e \approx 1/(6\alpha)$. Experimental validation (Section~\ref{sec:exp}) shows that query latency increases significantly for high aspect ratio filters.

\subsection*{A.5 Example: High Aspect Ratio Rectangle}

Consider a 2D rectangular filter with sides $100 \times 10$ (aspect ratio $\alpha = 10$) in a metadata space of side length $S = 1000$. With optimal layer selection where $w_{\ell^*} \approx 100$, the rectangle intersects at most $3 \times 2 = 6$ cubes ($3$ along the long dimension, $2$ along the short dimension since $r_{\min} = 10 < w_{\ell^*}$). The rectangle has area $1000$, and the union of $6$ cubes has area approximately $6 \times 100^2 = 60{,}000$. The elastic factor is $e \approx 1000 / 60{,}000 \approx 0.017$, much lower than the $\pi/36 \approx 0.087$ bound for circular filters at the same layer. Although the cube count remains a small constant, the shrinking elastic factor explains why high aspect ratio filters lead to poor query performance in \CG{}.

\begin{figure}[!t]
\centering
\begin{footnotesize}
\begin{tikzpicture}\vgap
    \begin{customlegend}[legend columns=4,
        legend entries={$\alpha=1$,$\alpha=4$,$\alpha=16$,$\alpha=32$},
        legend style={at={(0.5,1.15)},anchor=north,draw=none,font=\scriptsize}]
    \addlegendimage{line width=0.15mm,color=navy,mark=square,mark size=0.5mm}
    \addlegendimage{line width=0.15mm,color=amaranth,mark=o,mark size=0.5mm}
    \addlegendimage{line width=0.15mm,color=orange,mark=triangle,mark size=0.5mm}
    \addlegendimage{line width=0.15mm,color=forestgreen,mark=diamond,mark size=0.5mm}
    \end{customlegend}
\end{tikzpicture}
\\
\subfloat[SIFT 2D(filter ratio 0.1)]{\vgap
\begin{tikzpicture}[scale=0.85]
\begin{axis}[
    height=\columnwidth/2.60,
    width=\columnwidth/1.80,
    xlabel=recall@20(\%),
    ylabel=Qps,
    ymode=log,
    xmin=85,
    xmax=100.2,
    label style={font=\scriptsize},
    tick label style={font=\scriptsize},
    title style={font=\scriptsize},
    ymajorgrids=true,
    xmajorgrids=true,
    grid style=dashed,
]
\addplot[line width=0.15mm,color=navy,mark=square,mark size=0.5mm]
plot coordinates {
    (87.03, 4232.83)
    (95.36, 2464.99)
    (97.88, 1775.91)
    (98.85, 1388.96)
    (99.31, 1162.23)
    (99.57, 1000.76)
};
\addplot[line width=0.15mm,color=amaranth,mark=o,mark size=0.5mm]
plot coordinates {
    (92.72, 2480.41)
    (97.88, 1423.48)
    (99.13, 1015.06)
    (99.56, 798.07)
};
\addplot[line width=0.15mm,color=orange,mark=triangle,mark size=0.5mm]
plot coordinates {
    (93.97, 2273.80)
    (98.40, 1283.75)
    (99.35, 855.40)
    (99.67, 725.14)
};
\addplot[line width=0.15mm,color=forestgreen,mark=diamond,mark size=0.5mm]
plot coordinates {
    (95.51, 1711.42)
    (98.90, 960.64)
    (99.59, 684.47)
};
\end{axis}
\end{tikzpicture}
}
\hspace{2mm}
\subfloat[YFCC 2D(filter ratio 0.1)]{\vgap
\begin{tikzpicture}[scale=0.85]
\begin{axis}[
    height=\columnwidth/2.60,
    width=\columnwidth/1.80,
    xlabel=recall@20(\%),
    ylabel=Qps,
    ymode=log,
    xmin=95,
    xmax=100.2,
    label style={font=\scriptsize},
    tick label style={font=\scriptsize},
    title style={font=\scriptsize},
    ymajorgrids=true,
    xmajorgrids=true,
    grid style=dashed,
]
\addplot[line width=0.15mm,color=navy,mark=square,mark size=0.5mm]
plot coordinates {
    (95.77, 364.89)
    (98.57, 245.07)
    (99.39, 224.20)
    (99.64, 197.50)
};
\addplot[line width=0.15mm,color=amaranth,mark=o,mark size=0.5mm]
plot coordinates {
    (97.20, 84.56)
    (99.17, 68.97)
    (99.63, 59.39)
};
\addplot[line width=0.15mm,color=orange,mark=triangle,mark size=0.5mm]
plot coordinates {
    (98.07, 15.13)
    (99.40, 11.70)
    (99.75, 10.38)
};
\addplot[line width=0.15mm,color=forestgreen,mark=diamond,mark size=0.5mm]
plot coordinates {
    (98.11, 13.65)
    (99.45, 10.91)
    (99.76, 9.74)
};
\end{axis}
\end{tikzpicture}
}

\vgap\caption{Impact of filter aspect ratio on CubeGraph search efficiency (filter ratio 0.1). As aspect ratio increases, QPS drops dramatically while recall stays high due to the degrade of elastic factor.}
\label{fig:high-aspect-ratio}\vspace{-3ex}
\end{footnotesize}
\end{figure}
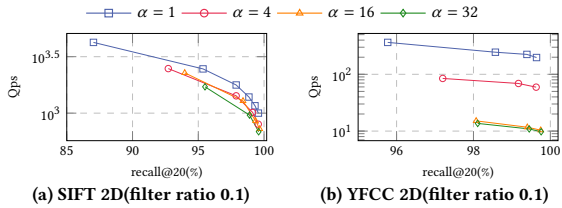

Fig.~\ref{fig:high-aspect-ratio} empirically validates this analysis on SIFT and YFCC (2D, filter ratio 0.1, 32 threads) for aspect ratios $\alpha \in \{1, 4, 16, 32\}$. On SIFT, throughput at 99.5\% recall@20 falls from $\sim$1000 Qps at $\alpha = 1$ to $\sim$684 Qps at $\alpha = 32$ (a $1.5\times$ slowdown). On YFCC, the effect is far more pronounced: $\sim$197 Qps at $\alpha = 1$ collapses to $\sim$9.7 Qps at $\alpha = 32$, a $20\times$ slowdown. The dataset gap is consistent with our theoretical model: the elastic factor decays as $\Theta(1/\alpha^{m-1})$, but the constant depends on the metadata distribution. Recall stays high in both cases, confirming that the slowdown comes from the elastic factor, not from the graph search structure.

\subsection*{A.6 Fly-Merge vs Cube-Merge}
\begin{figure}[!t]
\centering
\begin{footnotesize}
\begin{tikzpicture}
    \begin{customlegend}[legend columns=2,
        legend entries={Cube-merge4,Fly-Merge4},
        legend style={at={(0.5,1.15)},anchor=north,draw=none,font=\scriptsize,column sep=0.5cm}]
    \addlegendimage{line width=0.15mm,color=navy,mark=triangle,mark size=0.5mm}
    \addlegendimage{line width=0.15mm,color=forestgreen,mark=square,mark size=0.5mm}
    \end{customlegend}
\end{tikzpicture}

\begin{tikzpicture}[scale=0.9]
\begin{axis}[
height=\columnwidth/2.50,
width=\columnwidth/1.80,
xlabel=recall@20(\%),
ylabel=Qps,
xmin=90,
xmax=100.2,
label style={font=\scriptsize},
tick label style={font=\scriptsize},
title style={font=\scriptsize},
title style={yshift=-2.5mm},
ymajorgrids=true,
xmajorgrids=true,
grid style=dashed,
]
\addplot[line width=0.15mm,color=navy,mark=triangle,mark size=0.5mm]
plot coordinates {
(92.21, 1521.08)
(97.76, 852.62)
(99.07, 608.35)
(99.54, 478.69)
};
\addplot[line width=0.15mm,color=forestgreen,mark=square,mark size=0.5mm]
plot coordinates {
(90.82, 1413.12)
(97.31, 715.08)
(98.88, 561.51)
(99.44, 441.75)
(99.68, 367.10)
};

\end{axis}
\end{tikzpicture}\hspace{0.5mm}
\begin{tikzpicture}[scale=0.9]
\begin{axis}[
height=\columnwidth/2.50,
width=\columnwidth/1.80,
xlabel=recall@100(\%),
ylabel=Qps,
xmin=90,
xmax=100.2,
label style={font=\scriptsize},
tick label style={font=\scriptsize},
title style={font=\scriptsizes},
title style={yshift=-2.5mm},
ymajorgrids=true,
xmajorgrids=true,
grid style=dashed,
]
\addplot[line width=0.15mm,color=navy,mark=triangle,mark size=0.5mm]
plot coordinates {
(92.21, 1493.63)
(97.76, 837.57)
(99.07, 598.22)
(99.54, 456.37)
};
\addplot[line width=0.15mm,color=forestgreen,mark=square,mark size=0.5mm]
plot coordinates {
(90.82, 1320.04)
(97.31, 790.96)
(98.88, 565.50)
(99.44, 428.87)
(99.68, 369.72)
};

\end{axis}
\end{tikzpicture}\hspace{0.5mm}

\caption{Recall vs Qps comparison on SIFT dataset ($k=20$: left, $k=100$: right).
We compare hnsw-cube-merge4 and hnsw-fly-merge4 on the SIFT dataset.}
\label{fig:AExp-1}\vspace{-4ex}
\end{footnotesize}
\end{figure}
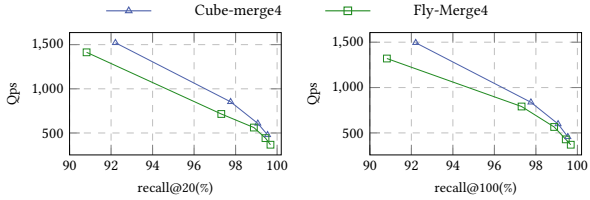
We evaluate the effectiveness of our two graph merging strategies: \emph{Fly-Merge} (on-the-fly dynamic discovery) and \emph{Cube-Merge} (predetermined cube identification) with 4 cube graph indices merged. Fig.~\ref{fig:AExp-1} presents the performance on the SIFT dataset with $k=20$ (left) and $k=100$ (right). Cube-Merge consistently outperforms Fly-Merge across all recall levels, achieving up to $1.4\times$ higher throughput at comparable recall. This advantage stems from Cube-Merge's upfront cube identification, which eliminates the dynamic discovery overhead during search. Fly-Merge incurs additional predicate evaluations and bitmap updates for each discovered cube, resulting in lower throughput despite its flexibility for complex filter shapes.

\end{document}